\documentclass[acmtosn]{acmtrans2m}
\usepackage{cite}
\usepackage[cmex10]{amsmath}
\usepackage{amssymb}
\usepackage{epsfig}
\usepackage{array}
\usepackage[font=footnotesize]{subfig}

\renewcommand{\thetable}{\arabic{table}}
\renewcommand{\thesection}{\arabic{section}}
\newtheorem{theorem}{Theorem}
\newtheorem{lemma}{Lemma}

\newtheorem{corollary}{Corollary}
\newtheorem{definition}{Definition}
\title{Analysis of Power-aware Buffering Schemes in Wireless Sensor Networks}
\author{Yibei~Ling, Chung-Min~Chen \\
Applied Research, Telcordia Technologies \\
Shigang~Chen \\
Department of Computer \& Information of Science \& Engineering, University of Florida
}
\begin{abstract}
We study the power-aware buffering problem in 
battery-powered sensor networks, 
focusing on the fixed-size and fixed-interval buffering schemes. 
The main motivation is to address
the yet poorly understood size variation-induced effect on
power-aware buffering schemes.
Our theoretical analysis elucidates the
fundamental differences between the fixed-size and 
fixed-interval buffering schemes in the presence of data size
variation. It shows that data size variation has
detrimental effects on the power expenditure of the fixed-size
buffering in general, and reveals that 
the size variation induced effects can be 
either mitigated
by a positive skewness or promoted by a negative skewness in size
distribution. By contrast, the
fixed-interval buffering scheme has an obvious 
advantage of being eminently immune to 
the data-size variation. Hence the fixed-interval buffering scheme is
a risk-averse strategy for 
its robustness in a
variety of operational environments. 
In addition, based on the 
fixed-interval buffering scheme, 
we establish the power consumption
relationship  
between child nodes and parent node in a static 
data collection tree, and give an in-depth analysis of the impact of 
child bandwidth distribution 
on parent's power consumption.

This study is of practical significance: it sheds new light on the
relationship among power consumption of buffering schemes, power
parameters of radio module and memory bank, data arrival rate and
data size variation, thereby providing well-informed guidance in determining an
optimal buffer size (interval) to maximize the operational
lifespan of sensor networks.
\end{abstract}
\begin{document}
\begin{bottomstuff}
Authors'addresses: Yibei Ling and Chung-min Chen, Applied Research,
Telcordia Technologies, 1 Telcordia Dr., Piscataway, NJ, 08854 
94301.\newline
Shigang Chen, Department of Computer \& Information of Science \& Engineering,
University of Florida, Gainesville, FL 32611.
\end{bottomstuff}
\maketitle
\section{Introduction}
A dramatic rise in research interest in power-conscious computing is
attributed, in part, to the growing awareness of the greenhouse
effect brought about by exponentially increasing number of computing
devices [\citeNP{Xie2008};\citeNP{Satyanarayanan1996}]. It is also driven by the impetus
to meet the long-duration operational requirement of battery-powered
sensor networks [\citeNP{Mainwaring2002};\citeNP{Woo2003};\citeNP{Culler2004};\citeNP{Gupta2003}]. This work is
motivated by problems arising from power-aware computing in general
and by battery-based sensor networks in particular.

A sensor network could be comprised of hundreds to thousands of tiny
sensor nodes. Each sensor node typically comprises a couple of
sensors, memory banks, a radio, and a microcontroller
[\citeNP{Hempstead2005}], being equipped with a stripped-down
version of the operating system. The sensor node can perform 
some basic computational tasks such as data measurement, filtering, 
aggregation, transmission/reception, and packet routing. Once deployed in the field, sensor nodes can
self-organize into a perceptive network that enables novel ways to
respond to emergencies, habitat monitoring and around-clock
environmental surveillance. 
The sensor nodes are required to autonomously operate under
harsh conditions for several months, even years, without
human intervention and maintenance [\citeNP{Mainwaring2002}]. In certain
cases, battery replacement or recharge may not even be possible
[\citeNP{Mainwaring2002};\citeNP{Landman1995}].
Thus the
premise of sensor networks to detect rarely-occurring events or to
monitor chronically changing events largely depends on the
lifespan of the sensor network. 
A review of essential features required by
sensor-based network applications yields a long list:
resilience, fault-tolerance, self-organization, and autonomy.
Despite such a rich feature set, the core requirement of
the sensor network is power conservation.

In a drive to bring power-aware computing to fruition, research
efforts have proceeded along three distinct yet closely related
tracks: 1) battery technologies;
2) hardware-based technologies; and
3) software-based technologies. Among these technologies, battery
technologies appear to be self-contained. Hardware-based and
software-based technologies are sharply distinct but mutually
dependent as well.

The power conservation requirement fundamentally
reshapes how hardware modules should be designed,
implemented and assembled. 
The evolution of the hardware-based approach
is a process of continuously replacing power-inefficient components with
ultra-low power modules, and substituting general-purpose
components with specially designed power efficient ones.
Wireless radio and memory components 
have long been recognized as the
biggest power spenders in a sensor node system
[\citeNP{MinLee2007}].  
Realization of
this shortcoming has directed research attention toward designing
ultra-low power radio and 
memory components with
multi-power mode capability [\citeNP{MinLee2007};\citeNP{Flautner2002}].
However, the hardware multi-power mode capability
alone does not warrant power efficient computing in practice.
The reason is that a transition between operating power modes (from a
low-power mode to a high-power mode or vice versa) bears a {\em
resynchronization cost}, {\em i.e.}, a certain amount of 
energy
incurred to demote or to elevate the operating power level. As a
result, the availability of hardware multi-power mode capability presents a new set of
collateral risks of being misused: a blind choice of operating
power mode might incur an excessive transition cost that
neutralizes the benefits brought out by power-aware hardware
design.

To reap the benefit of multi-power mode feature in a hardware
design, software-based technologies are concerned with the
design of algorithms/protocols
that can exploit the
multi-power capability, thus serving as a reinforcer to the
hardware-based power-aware technologies. The whole idea underlying
the software-based approaches centers around the exploitation of
quiescence in workload, linking the power mode of
a component to its workload characteristics.

\citeN{MinLee2007} introduced a power-aware
buffer cache management scheme called {\em PABC} for compressing and
migrating 
active pages in both user-space and kernel-space onto a
few memory units. Their experimental study indicated that the 
{\em PABC}
scheme can reduce the energy consumption of the buffer cache
by an impressive $63\%$. 
\citeN{Flautner2002} observed that in practice
the hot (active) cache only 
accounts for a small subset of 
on-chip caches for most of time. This observation leads to
an architectural design that 
exploits such a workload pattern
to place the cold cache into drowsy mode, 
thereby saving a substantial power consumption. 
The experimental studies showed that 
about the $80\%-90\%$ of cache 
can be maintained in a drowsy (idle) mode
without affecting performance by more than
$1\%$. \citeN{Ling2007}, on the other
hand, derived closed form optimal buffering strategies, under the
condition that the received data size is entirely devoid of
variability and identical to the size of a memory bank. This
assumption greatly simplifies mathematical derivation. It, 
however, appears to be inadequate in capturing the
essence of sensor networks in a realistic setting. 

It is widely recognized that 
idle listening is the major energy spender
in senor networks. For example, 
experimental study shows that
$99\%$ of energy is dissipated on idle listening 
if a node is always turned on [\citeNP{Lin2005};\citeNP{Shnayder2004}].
Many power management 
protocols are proposed to reduce power consumption
in listening.  
Asynchronous low power 
listening (ALPL) scheme 
uses duty cycling to reduce the listening energy.
A node is required to periodically wake up and check the radio channel.  
In general, the energy saving on listening
at receiver is at the expense of sender,
because the sender must open
the radio channel long enough to ensure 
correct message reception. 
Synchronous Low Power Listening scheme (SLPL) 
improves on the ALPL scheme in its ability to
coordinate sender's transmit mode
with the receiver's periodic check [\citeNP{Ye2002};\citeNP{Jurdak2007}].
The weakness of SLPL is that it demands a 
high-quality time synchronization 
among a group of nodes, which incurs a non-negligible 
amount of energy. In addition, 
the design of an energy efficient
wake-up/sleep protocol 
is often application dependent and complicated
in practice. 
Hence, it is hard to design a general power 
management system based on 
wake-up/sleep scheduling. 

A radically different approach, called radio-triggered wake-up 
power management, is proposed 
by [\citeNP{Lin2004};\citeNP{Lin2005};\citeNP{Ansari2009}]. 
It uses a radio-triggered circuit as one interrupt input of the processor. 
The circuit itself does not require any power supply
and is powered by the radio signals themselves. As a result,
the radio-triggered power 
scheme allows nodes to sleep without need for
periodic wake-up to check channel signals, 
thereby completely eliminating listening power consumption. 

In this paper we study the 
power-aware buffering problem by exploiting
the multi-power mode in radio and memory components
and the radio-triggered power 
management. The optimization
objective is to minimize power consumption in the context of two
buffering paradigms: 
the well-known {\it fixed-size} and the 
lesser-known {\it
fixed-interval} buffering schemes. 
In particular, we 
focus on the size variability-induced effect on 
these power-aware buffering schemes.

To the best of our knowledge, the effect of size variability on 
power consumption of buffering schemes has not been addressed before.
Our analysis provides insight into 
the poorly understood effect of size variability on the 
power-aware buffering schemes,
thereby providing a theoretical guidance for performance tuning in practice. 
The novelty of this paper is its adoption
of asymptotic analysis, which allows us to model the limitation of
power-aware buffering schemes without sacrificing simplicity and
elegance.

The remainder of this paper is organized as follows: Section~2
presents relevant definitions and
prerequisite theorems that facilitate derivation of the main
theorems. Section~3 presents the exposition of
theoretical analysis for both power-aware fixed-size and
fixed-interval buffering schemes. 
Section~4 compares the performance between the fixed-size and fixed-interval
buffering schemes in both the absence and presence of size variation.
Section~5 discusses the gain of power-aware buffering schemes over
power-oblivious ones in terms of power conservation,
with some examples to illustrate the effect of power-aware
buffering on the lifespan of sensor nodes.
the power consumption relationship between the parent
and child nodes in a data collection tree is presented.
Section~6 concludes this paper.

\section{Background}
Multi-power mode radio and memory components are the main hardware
prerequisites of power-aware buffering schemes in this paper. 
The efficacy of power-aware hardware
design relies on the ability of software-based approach to exploit
the potential of power-aware hardware design. 

\begin{figure}[htb]
\centerline{\psfig{file=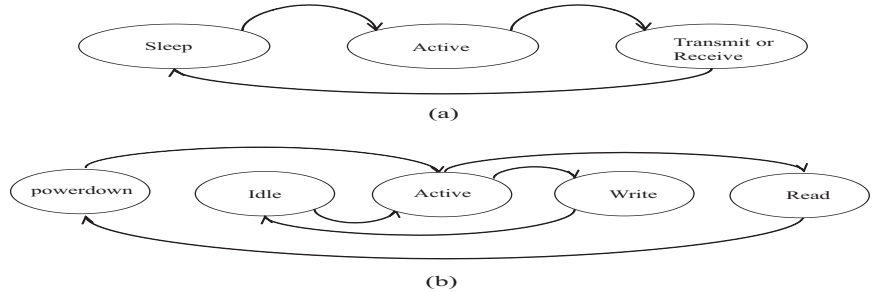,height=1.5in,width=4.5in}}
\caption{(a) Power state transition diagram of wireless radio module
(b) Power state transition diagram of memory bank } 
\label{fig:flow}
\end{figure}

To study the performance of power-aware buffering schemes, let's
discuss at length the power-mode transition pattern of radio and memory
components. We assume that nodes 
use the radio-triggered power management, 
thus do not 
incur listening power consumption. 

The power-mode of a multi-power radio component 
can be subsumed into 1) 
the sleep mode and 2) the active mode.
A sleep-mode radio inhibits data 
transmission/reception. 
An active-mode radio permits 
data transmission/reception but incurs
more power than when in sleep mode. 
To save power consumption,  
the radio is placed 
into sleep mode most of the time; 
it is only elevated to active mode 
(by a radio-triggered wakeup component)
when data transmission/reception
is needed. After completing data
transmission/reception, 
the radio is put back to sleep mode. 
The {\em sleep-active-transmit-sleep} 
transition diagram in
Figure\,\ref{fig:flow}(a) 
forms a typical power-aware radio working pattern.
 
A memory bank
refers to {\em the
minimum size of a memory unit whose power mode can be
independently altered} [\citeNP{Hempstead2005}].
Its power mode could be broadly classified into 
three categories: 1) the powerdown mode; 2) the idle mode;
and 3) the active mode.
A powerdown-mode memory bank means that 
the voltage supply to the memory bank is cut off,
resulting in a sizable reduction in 
current leakage [\citeNP{Flautner2002};\citeNP{Tarjan2006}]. 
The idle (sleep or drowsy) mode is 
the minimum power mode that 
preserves the 
stored information but inhibits 
writing and reading of data.   
An idle-mode or powerdown-mode memory bank 
must be reinstated to the active mode before a 
read/write operation can be performed. 
An active mode memory bank not only 
retains the stored information but also 
allows the data to be written/read.
The power consumption in a powerdown-mode memory bank 
is negligibly small. An idle-mode memory bank consumes less power
and has less functionality than
an active-mode memory bank.   

\begin{figure}[htb]
\centerline{\psfig{file=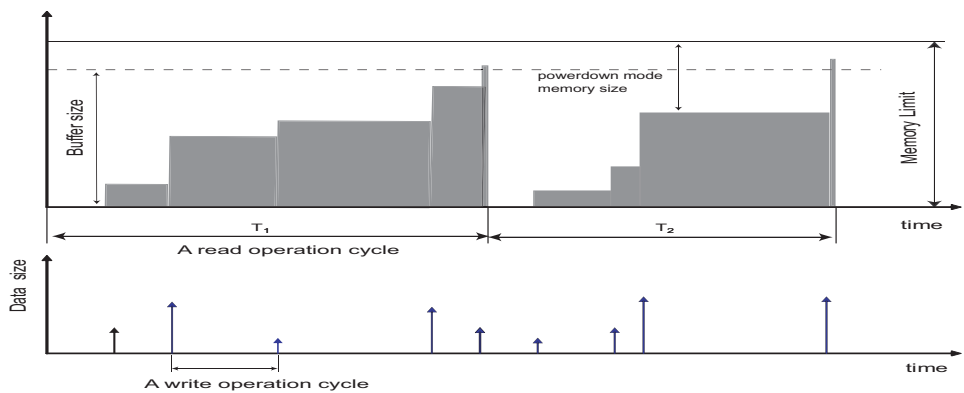,height=2in,width=5in}}
\caption{Evolution of powerdown-mode memory size (blue) and 
idle-mode memory size (red)} 
\label{fig:fl1}
\end{figure}

In an ideal power-saving data buffering scenario, 
the power-mode transition could be divided into 
1) the write power-mode transition and 2) 
the read power-mode transition. 
A write power-mode transition is initiated 
by an interrupt of radio-triggered circuit. 
The sensor node first powers up the memory banks from the powerdown mode to  
the active mode, and then writes data into the memory banks. 
After completing a {\em write}
operation, the involved memory banks are demoted to 
the idle mode to preserve power 
while retaining the stored information. 
The {\em powerdown-active-write-idle} power mode transition
in Figure\,\ref{fig:flow}(b) forms a {\em write} operation cycle. 

A read power-mode transition 
is initiated when a specified buffer
size or interval threshold is reached. 
Thus, a {\em read} power-mode transition may
comprise more than one {\em write} power-mode transition
cycle, depending on the specified buffer size. 
It involves loading,
transmitting and clearing up all the buffered data. 
To do so, it elevates the power mode of the buffered memory banks 
from idle mode to active mode. 
Once reading of data is completed,
the corresponding memory banks are put back to 
powerdown mode.
The {\em idle-active-read-powerdown} 
power mode transition
in Figure\,\ref{fig:flow}(b) forms a {\em read} 
operation cycle,
which is synchronized with the {\em sleep-active-transmit-sleep}
power mode transition in the radio. In other words, 
data transmission in the radio
is initiated immediately right after
reading/loading buffered data from the memory banks. 

The bottom graph in Figure\,\ref{fig:fl1} depicts the size of 
of arrival data (collected via sensors) as a function of time. The top graph shows the
evolution of powerdown-mode memory size (total memory size minus stair height) 
and of idle-mode memory size (stair height). Observe that   
the powerdown-mode 
memory size 
shrinks as the arrival data are 
accumulated in idle-mode memory banks.
The size of idle-mode 
memory banks grows in a stair-like fashion when it is less than
the prescribed buffer size. Once it hits the prescribed 
buffer size and a transmission of buffered data is initiated.  
This forms a {\em read} power-mode transition cycle for memory banks, 
as well as a radio power-mode transition cycle. In practice, 
the duration of the {\it
read} transition cycle may fluctuate widely: it 
could be very sensitive to 
the buffering policy, data arrival 
rate and data size distribution.
In order to analyze the power-aware buffering issues, we begin with two buffering policies as follows:
\begin{definition} \label{def:define1}
A buffering policy is said to be {\em stationary} if its decision
depends only on its current state and not on the time. A buffering
policy is said to be deterministic
if its decision is a {\em deterministic} function of the current state.
\end{definition}
\begin{enumerate}
\item Fixed-size buffering scheme: data transmission is initiated
immediately when a fixed (prescribed) buffer size is reached.
\item Fixed-interval buffering scheme: data transmission is commenced
periodically with a fixed time interval.
\end{enumerate}
 
There is a clear distinction between the well-known
fixed-size and lesser-known fixed-interval buffering schemes: 
the threshold of the {\it
fixed-size} buffering scheme depends on the 
size of the buffered data,
and that of the {\it fixed-interval} one 
depends on a specified 
time interval. 
By definition~\ref{def:define1}, the {\it fixed-size}
buffering scheme is stationary 
while the {\it fixed-interval} 
buffering
scheme is deterministic. For notational convenience, 
we use
the superscripts $FS$ and
$FI$ to denote the 
fixed-size and fixed-interval buffering schemes throughout
the paper.
Before dwelling into a detailed derivation, we 
introduce relevant notions and
essential prerequisites.
\begin{definition} \label{def:skewness}
Let $x$ be a random variable following a probability distribution
${\cal F}$, {\it i.e.}, $x\!\sim\!{\cal F}$, the skewness of 
$x$, denoted by $\gamma(x)$, is defined as
\begin{align} \label{equ:skewness}
\gamma(x)
=\dfrac{E[(x-\mu_x)^3]}{(E[x-\mu_x)^2])^{3/2}}
=\dfrac{\mu_3}{\sigma_x^3}.
\end{align}
\noindent The coefficient of variation 
of $x$, denoted by $c_v(x)$,
is defined as
\begin{align} \label{equ:vc}
c_v(x) =\dfrac{\sqrt{E[(x-\mu_x)^2]}} {E[x]}
=\dfrac{\sigma_x}{\mu_x},
\end{align}
where $E[]$ is the expected function and $\mu_x\!=\!E[x]$. 
\end{definition}
In probability theory, $\gamma(x)$ is 
the third standardized moment for measuring the
degree of asymmetry. It can be
further divided into positive ($\gamma(x)>0$) and negative skewness ($\gamma(x)<0$). 
The function $c_v(x)$
is a measure for the degree of dispersion. 

\begin{definition}
An integer-valued random variable $\mathbf{n}$ is said to be a {\em
stopping time} for the sequence $x_1,x_2,\cdots$ if the event
$\{\mathbf{n}=n\}$ is independent of $x_{n+1},x_{n+2},\cdots$ for
all $n=1,2,\cdots$.  
\end{definition}

\noindent Wald's equation:
Suppose {\em $y_1, y_2 \cdots$ are iid random variable with
finite expectation $E[y_i]\!=\!\mu_y$, and $\mathbf{n}$ is a
stopping time for $y_1,y_2\cdots$ such that
$E[\mathbf{n}]<\infty$, then
\begin{align} \label{equ:wald}
E[\sum_{i=1}^{\mathbf{n}} y_i]= 
E[\mathbf{n}]E[y]\,=\,\mu_y E[\mathbf{n}]
\end{align}}

The following theorem establishes the asymptotic behavior of
stopping time variance {\em w.r.t} buffer size $b$.  

\begin{theorem} \label{the:theosimple}
Let $\{x_i>0, i \geq 1 \}$ be a random positive walk (increment)
with mean of $\mu_x\!=\!E[x_i]\!>\!0$ and finite variance of $\sigma_x^2$.
Let stopping time $\tau(b) = \min \{ n \geq 1:
\sum^{n}_1 x_i > b\}$. When $b$ is sufficiently large,
the stopping time variance $\sigma_{\tau_b}^2$ becomes
\begin{align} \label{equ:simplekey}
\sigma_{\tau_b}^2 =
\dfrac{b\sigma_y^2}{\mu_y^3} + k^*
=\dfrac{b c_v^2(y)}{\mu_y} + k^*
\end{align}
where $k^*$ is expressed as
\begin{align} \label{equ:xkeysim}
k^*  =
\dfrac{5 c^4_v(y)}{4}
+\dfrac{1}{12}-\dfrac{2 c_v^3(x) \gamma(y)}{3},
\end{align}
where $c_v(y)$ denotes the coefficient of variation,
and $\gamma(y)$ the skewness of $y$. 
\end{theorem}

Proof of theorem~\ref{the:theosimple} is given in the
Appendix. 

Theorem~\ref{the:theosimple} states that
in an asymptotic sense, the stopping time variance is
linearly proportional to the buffer size $b$, with a proportionality
constant of $c^2_v(y)/\mu_y$ and the intercept $k^*$ 
determined by both $c_v(x)$ and $\gamma(y)$.
It means that $c^2_v(y)/\mu_y$ could play a central role in
determining the stopping time variance. The
magnitude of intercept ($k^*$) can be either mitigated by positive
skewness or augmented by negative skewness.
It is noteworthy that Theorem~\ref{the:theosimple} is a special
case of Lau's theorem [\citeNP{Lai1977};\citeyearNP{Lai1979}] 
under the positive random increment condition,
which results in a substantial simplification. The following
corollaries are  special cases of Theorem~\ref{the:theosimple} in
which the random walk (increment) is assumed to be exponentially or
Erlangly distributed. 

\begin{corollary} \label{coro:coro1}
For a given buffer size $b$, 
the stopping time variance $\tau(b)$
for an exponential random walk with 
mean $1/\lambda_{e}$ is
\begin{align} \label{equ:expo}
\sigma^2_{\tau(b)}(exp) = \lambda_{e} b,
\end{align}
where $\sigma_{\tau(b)}(exp)$ refers to the stopping time 
variance
{\em w.r.t.} an exponential random walk (increment). 
\end{corollary}

\begin{proof} For an exponential random walk, 
by Definition\,\ref{def:skewness}
we get
$c_v(y)=1,\gamma(y)=2$. Substituting
$c_v(y)$ and $\gamma(y)$ into (\ref{equ:simplekey})
leads to (\ref{equ:expo}). $\hfill$
\end{proof} 

\begin{corollary} \label{coro:coro2}
For a given buffer size $b$, the stopping time variance for an
Erlang random walk with parameters $(\alpha,\!\lambda_\alpha)$ is
\begin{align} \label{equ:erl}
\sigma^2_{\tau(b)}(erlang) = \frac{\lambda_\alpha b}{\alpha^2} +
\frac{1}{12}(1 -\frac{1}{\alpha^2}),
\end{align}
where $\alpha\!>\!1$ is the shape parameter (an integer),
$\lambda_\alpha$ refers to the rate, and  $\sigma_{\tau(b)}(erlang)$
refers to the stopping time variance
{\em w.r.t.} an Erlang random walk (increment).
\end{corollary}

\begin{proof} By Definition\,\ref{def:skewness}, 
for an Erlang random walk,
we obtain $c_v(y)\!=\!\frac{1}{\sqrt{\alpha}}$ and
$\gamma(y)\!=\!\frac{2}{\sqrt{\alpha}}$. Substitution of $c_v(y)$ and
$\gamma(y)$ into (\ref{equ:simplekey})
yields (\ref{equ:erl}).  $\hfill$
\end{proof} 
 
Consider the differential stopping time
variance between the exponential and 
the Erlang random walks by
subtracting (\ref{equ:expo}) with (\ref{equ:erl}).
\begin{align} \label{equ:compar1}
\sigma^2_{\tau(b)}(exp) -\sigma^2_{\tau(b)}(erlang) =
b\lambda_{e}-\frac{\lambda_\alpha b}{\alpha^2} +
\frac{1}{12}(1-\frac{1}{\alpha^2})
\end{align}
Notice that the mean increment size
of exponential random walk is $\mu_e = 1/\lambda_e$ and that of
the Erlang walk is $\mu_{\alpha} = 
\frac{\alpha}{\lambda_\alpha}$.
Letting $\mu_{\alpha}=\mu_e=\mu$, then (\ref{equ:compar1}) is reduced
to
\begin{align} \label{equ:comparison}
\sigma^2_{\tau(b)}(exp) -\sigma^2_{\tau(b)}(erlang) =
(1-\frac{1}{\alpha}) \left(\frac{b}{\mu}-\frac{1}{12}(1
+\frac{1}{\alpha})\right),
\end{align}
then
\begin{align} \label{equ:comp}
\sigma^2_{\tau(b)}(exp) -\sigma^2_{\tau(b)}(erlang)
\begin{cases}
> 0 & \frac{b}{\mu}>\frac{1}{12}(1 +\frac{1}{\alpha})\\
\leq 0 & \frac{b}{\mu} \leq \frac{1}{12}(1 +\frac{1}{\alpha}).
\end{cases}
\end{align}

(\ref{equ:comp}) means that with the same
mean increment size, the stopping time under the exponential random
walk (increment) has a wider variance than that under the Erlang
walk as long as the condition $b> \mu/12$ is met. 

Consider a hyper-exponential random walk (increment) as
$\sum_{i=1}^2\!p_i\lambda_i\!\exp{(-\lambda_i x)}$ where
$\sum_{i=1}^2\!p_i\!=\!1$ (letting $p_1=p, p_2=1-p$). The differential stopping
time variance between the hyper-exponential and the exponential
walks is $\sigma^2_{\tau(b)}(hp)\!-\!\sigma^2_{\tau(b)}(exp)$.
Under the same mean increment, it becomes
\begin{align} \label{equ:compare10}
\sigma^2_{\tau(b)}(hp)-\sigma^2_{\tau(b)}(exp)=
\frac{b(c^2_v(y)-1)}{\mu} + k^*,
\end{align}
where $k^*$ is explicitly given in
(\ref{equ:simplekey}). This implies that the differential stopping
time variance is linearly proportional to the buffer size $b$, that
is, $\sigma_{\tau(b)}^2(hp) -\sigma_{\tau(b)}^2(exp) 
\propto \frac{b}{\mu}(c^2_v(y)-1)>\!0$ when $b$ is sufficiently large. 
Namely, in an asymptotic sense, the hyperexponential random walk has a wider
variance in the stopping time than 
the exponential walk under the
same mean increment size condition.

Consider the {\it fixed-size buffering} scheme with a
size of $b$.
Define the stopping time, denoted by $\tau(b)$, to be a random
variable that takes on values in $[0,\infty)$. 
One sees that $\tau(b)$ is a function of $b$ and 
the size distribution of the data $\{y_i>0: i \geq 0\}$:
\begin{align} \label{equ:def}
\tau(b)=\min\{ n: \sum_{i=1}^{n} y_i \geq b \},
\end{align}
where $\tau(b)$ is referred to as the 
{\em first ladder epoch} and
$\sum_{i=1}^{\tau(b)} y_i$ is called
the {\em first epoch height} 
[\citeNP{Lai1977};\citeyearNP{Lai1979};\citeNP{Feller1971}].

\begin{table}[bht]
\centering
\begin{tabular}{ll}
\multicolumn{2}{c}{Data Traffic}  \\ \hline
$\lambda$ &  Poisson data arrival rate \\ 
$\mu_y$ & mean value of data size distribution \\ 
$b_{size}$ & size of a memory bank \\ 
$\lambda \mu_y$ & bandwidth  \\ \hline
\multicolumn{2}{c}{Radio Module} \\ \hline
$e^{wu}_w$ & energy for a radio wakeup \\ 
$e^{RX}_w$ &  energy for one-byte reception \\ 
$e^{TX}_w$ &  energy for one-byte transmission ($e^{RX}_w \approx e^{TX}_w$)\\  \hline
\multicolumn{2}{c}{Memory Bank} \\ \hline
$p^{idle}_m$ & power of idle state of one memory bank \\ 
$e^{ena}_m$ & energy to elevate from powerdown to active \\ 
$e^{dem}_m$ & energy to demote from active to idle\\
$e^{r}_m$ & energy of reading one byte \\ 
 $e^{w}_m$ & energy for writing one byte \\ 
$e^{resyn}_m$ & $(e^{ena}_m + e^{dem}_m)/2$ \\ \hline
\end{tabular}
\caption{Symbols and Meanings} \label{tab:tab4}
\end{table}

One key step is to establish a relationship
between the mean
stopping time (the first ladder epoch) and
the mean size of the data distribution. Assume
that the sensor node has enough buffer capacity to accommodate 
{\em first ladder height} (overshoot) with respect to the buffer size $b$. 

\begin{theorem} \label{the:the1}
Let $\{y_i\!>\!0,\,i\!\geq\!0\}$ be the sequence of increment sizes
with mean $\mu_y$, and $b$ be the buffer size, the mean stopping
time $E[\tau(b)] \approx \frac{b}{\mu_y}$. 
\end{theorem}

\begin{proof} By Wald's equation in (\ref{equ:wald}) we
obtain the relation $
\sum_{i=1}^{\tau(b)-1} y_i
< b \leq \sum_{i=1}^{\tau(b)} y_i$. 
Taking expectation on both sides of this relation yields
\begin{align} \label{equ:ooo1}
E [\sum_{i=1}^{\tau(b)-1} y_i] 
 < b \leq  E[\sum_{i=1}^{\tau(b)} y_i]  
\Longrightarrow  (E[\tau(b)]-1) \mu_y 
 < b \leq E[\tau(b )] \mu_y.
\end{align}
Dividing both sides of (\ref{equ:ooo1}) by $\mu_y$ 
completes the proof.  $\hfill$
\end{proof} 
The preceding theorem asserts that the {\it fixed-size} buffering
scheme with a buffer size of $b$ can hold
$\frac{b}{\mu_y}$ data packets on average when the data size is
randomly distributed with a mean of $\mu_y$, which 
is in line with our intuition.  

For the sake of clarity, 
we summarize the power parameters in 
Table\,\ref{tab:tab4}. 
The subscripts $m$ and $w$ denote the
memory bank and radio module.
$e^{ena}_m$ and $e^{dem}_m$ refer to the
energy required to elevate a
powerdown-mode memory bank to active 
mode, and to demote an active-mode 
memory bank to idle mode,
$e^{resyn}_m$ is a resynchronization cost being 
equal to the mean value of   
$e^{ena}_m$ and $e^{dem}_m$,
and $\lambda \mu_y$ the data volume per time
unit, termed as {\em bandwidth}
due to conceptual similarity.
Since the duration of an active-mode memory bank
is extremely short, thus the energy 
consumed in the active-mode
could be reasonably ignored. 
Similarly, the energy consumed by the active-mode 
of a radio module is outweighed by 
$e^{TX}_w,e^{RX}_w$, and hence is ignored.

\section{Power-aware Buffering Schemes}
\subsection{Fixed-Size Buffering Scheme}
In this subsection, we consider the {\it fixed-size buffering}
scheme under randomly distributed 
data size with Poisson arrival. 
Assume that data size follows a 
certain probability
distribution with a finite mean of $\mu_y$. Let
$(x_i,y_i), i \geq 0$ be a sequence of random vectors in which
$\{x_i, i \geq 0\}$ refers to a 
random variable denoting
the interarrival times of Poisson arrival 
data and
$\{y_i, i \geq 0\}$ be a random variable representing the
size of the arrival data. The random variables $x_i$ and $y_i$ are assumed
to be
mutually independent. 

\begin{theorem} \label{theo:the3}
Let $\lambda$ be a Poisson arrival rate, $\mu_y$ be the mean data size, $b_{size}$ be the
size of a memory bank, $e^{wu}_w$ be the per radio wakeup
energy, and $p^{idle}_m$ be the idle-mode power consumption of a
memory bank. Then the optimal buffer size $b^*$ for the {\it
fixed-size} buffering scheme is
\begin{align} \label{equ:op-fixed-size}
b^*= \sqrt{\frac{2 b_{size} e^{wu}_w \lambda \mu_y}{p^{idle}_m}
+ \mu_y^2 k^*}
\end{align}
\end{theorem}
where $k^*$ is given in (\ref{equ:simplekey}). 

\begin{proof} Consider a random vector sequence
$(x_i,y_i)$, $i \geq 0$, where $\{x_i, i \geq 0\}$ represents the
arrival time instants (Poisson arrival) and $\{y_i, i \geq 0\}$ is
a sequence of received data sizes, with a mean $E[y_i]=\mu_y$ and
a variance $\sigma^2_y$.

Define a renewal reward process [\citeNP{Ross1996}] with the cycle
length being equal to the time duration of stopping time $\tau(b)$
as follows:
\begin{equation}
L_c =
\sum\limits_{i=0}\limits^{\tau(b)}x_{i+1}-x_{i},
\end{equation}
where $L_c$ denotes the length of 
a renewal cycle, 
and $x_{i+1}\!-\!x_i,i \geq 0$ is interarrival times. 
Letting $s_k = \sum^{k}_{i=1} x_i$. 
Thus the total energy
$e^{FS}(b)$ over a renewal cycle is
\begin{align} \label{equ:total_energy_fixed_size}
e^{FS}(b) &=e^{wu}_w+ \dfrac{p^{idle}_m}{b_{size}}
\sum_{i=1}^{\tau(b)}(s_{\tau(b)}-s_i)y_i 
+\sum_{i=1}^{\tau(b)} 
(e^{TX}_w yi) +
\sum_{i=1}^{\tau(b)} (e^{RX}_w y_i +e^{wu}_w) \\ \nonumber
& +
\sum_{i=1}^{\tau(b)} 
(e^w_m+e^r_m)y_i+2e^{resyn}_m  
\end{align}
Let us return to explaining each term in
(\ref{equ:total_energy_fixed_size}). 
$\frac{p^{idle}_m}{b_{size}} \sum_{i=1}^{\tau(b)}
(s_{\tau(b)}-s_i) y_i$ denotes the accumulated idle-mode energy for
the number of memory banks in a renewal cycle, and
$\sum_{i=1}^{\tau(b)}e^{TX}_w y_i$ 
refers to the transmission energy, 
$\sum_{i=1}^{\tau(b)} (e^w_m+e^r_m) 
y_i+2e^{resyn}_{m}$ is the total energy required to
write/read data into/from the memory banks, plus the
resynchronization energy, and 
$\sum_{i=1}^{\tau(b)} e^{RX}_w y_i +e^{wu}_w$
refers to the total energy for receiving data, plus 
the energy for radio wakeup for receiving data. 
The term $e^{wu}_w$ refers to per radio wakeup energy 
for data transmission. In other words, in each 
renewal cycle, the transmission radio wakeup occurs 
only once, while the reception radio wakeup occurs
$\tau(b)$ times. 
Recall that we assume that nodes use 
the radio-triggered power management scheme, 
thereby the radio wake-up can be initiated  
without incurring listening energy.
By Wald's equation, we get
\begin{align}
E[L_c]
=E\left[\sum_{i=0}^{\tau(b)}x_{i+1}-x_{i}\right]=E[\tau(b)]E[x_{i+1}-x_{i}]\nonumber
 =\frac{E[\tau(b)]}{\lambda} =\frac{b}{\lambda \mu_y}
\end{align}
Define $e^{FS}(t)$ to be the accumulated energy consumption at time
$t$, where multiple renewal cycles may have occurred in the time
period $[0,t]$. 
By the renewal reward theory [\citeNP{Ross1996}],
the long-run mean average energy consumption 
is  
\begin{align} \label{equ:tt1}
\overline{e^{FS}(b)} \stackrel{\text{\tiny def}}{=}
\lim_{t\rightarrow \infty} \dfrac{e^{FS}(t)}{t}
= \dfrac{E[e^{FS}(b)]} {E[L_c]} =
\dfrac{E[ e^{FS}(b)]}{\frac{b}{\lambda \mu_y}},
\end{align}
where the unit of $\overline{e^{FS}(b)}$
is the watt (W), rather than the joule (J).
Letting $e^{TX}_w =e^{RX}_w$. 
Taking expectation of the third term in
(\ref{equ:total_energy_fixed_size}) and 
applying Theorem~\ref{the:the1} give
\begin{align} \label{equ:thirdterm}
& E\left[\sum\limits_{i=1}\limits^{\tau(b)}
(2e^{TX}_w+e^w_m+e^r_m)y_i+ e^{wu}_w+2e^{resyn}_m\right] \nonumber \\ 
&= E[\tau(b)]E[ (2e^{TX}_w+e^w_m+e^r_m)y_i +e^{wu}_w+2e^{resyn}_m ]
\nonumber \\
&=\dfrac{b}{\mu_y}
\left(\mu_y(2e^{TX}_w+e^w_m+e^r_m)+ e^{wu}_w +2e^{resyn}_m \right)
\end{align}
It follows from (\ref{equ:simplekey}) 
and the assumption of independence of
$y_i$ and $s_i=\sum_{j=1}^i x_j$,  
the expectation of the second term in
(\ref{equ:total_energy_fixed_size}) thus becomes 
\begin{align} \label{equ:tt2}
E\left[\sum\limits_{i=1}^{\tau(b)}(s_{\tau(b)}-s_i)y_i\right]
& =\mu_y E\left[\sum\limits_{i=1}
\limits^{\tau(b)} s_{\tau(b)}-s_i\right]=
\dfrac{\mu_y (E[\tau^2(b)]-E[\tau(b)])}
{2\lambda} \\ \nonumber 
& =
\frac{\mu_y \left(E^2[\tau(b)]+\sigma^2_{\tau(b)}-E[\tau(b)]\right)}
{2\lambda}=\dfrac{\dfrac{b^2}{\mu_y}+b(c^2_v(y)-1)+ \mu_y k^*}{2\lambda},
\end{align}
where $k^*$ is given
in (\ref{equ:xkeysim}). Substitution of 
(\ref{equ:thirdterm})-(\ref{equ:tt2})
into (\ref{equ:tt1}) yields
\begin{align} \label{equ:fixed-size}
\overline{e^{FS}(b)} & = \dfrac{\lambda \mu_y e_w^{wu} +
\frac{p_m^{idle} \mu_y^2 k^*}{2 b_{size}}}{b}
 +\lambda \left( \mu_y (2e^{TX}_w+e^w_m+e^r_m)+
e^{wu}_w+2e^{resyn}_m \right) \nonumber \\   
&+ \frac{p^{idle}_m \mu_y}{2 b_{size}}
(c^2_v(y)-1) + \frac{p_m^{idle} b}{2b_{size}}
\end{align}
Solving $ \dfrac{\partial \overline{e^{FS}(b^*)}}{\partial
b^*}=0$ yields
$ b^*=  \sqrt{\frac{2 e^{wu}_w b_{size} \lambda \mu_y}
{p^{idle}_m}+ k^* \mu_y^2 }$. 
To prove that $b^*$ is the optimal buffer size, it suffices to
show that
\begin{align} \label{equ:ttt4}
\lim_{b \rightarrow b^*} \frac{\partial^2
\overline{e^{FS}(b)}}{\partial b^2}  = 2\dfrac{e^{wu}_w
\lambda \mu_y + \dfrac{p_m^{idle} k^* \mu_y^2}{2 b_{size}} }{(b^*)^3} > 0.
\end{align}
The proof is thus completed. $\hfill$
\end{proof} 

The following corollary is a special case of
Theorem~\ref{theo:the3}.

\begin{corollary}
When the size of received data is constant and 
identical to that of a memory
bank, the optimal size of power-aware {\em fixed-size
buffering} is expressed as
\begin{align} \label{equ:tttt}
n^*= \sqrt{\dfrac{2 \lambda e^{wu}_w}
{p^{idle}_m}+\dfrac{1}{12}} \approx
 \sqrt{\dfrac{2 \lambda e^{wu}_w}
{p^{idle}_m}},
\end{align}
\end{corollary}
where $n^*$ in (\ref{equ:tttt}) refers to
the number of memory banks used, hence the optimal
buffer size is $b^*= n^* b_{size}$ ($b^*$ is a multiple of 
memory bank size $b_{size}$).  

\begin{proof} It is almost trivial and therefore omitted. $\hfill$
\end{proof} 

Theorem~\ref{theo:the3} takes into account the
impact of unevenly distributed data size, thereby 
generalizing the previous work [\citeNP{Ling2007}]
beyond the fixed-size data condition. 
It shows that the first two moments of 
data size distribution (mean and variance) alone are not sufficient
to capture the dynamics of the power-aware fixed-size buffering.
The term $\mu_y^2\,k^*$ in (\ref{equ:op-fixed-size}) represents
the impact of varying-size data on the power-aware {\it fixed-size}
buffering scheme, which is orthogonal to the data arrival rate
$\lambda$. Such an impact can be quantitatively isolated in the form
as:

\begin{align} \label{equ:tttt5}
\Delta_{v} b\!=\!\sqrt{\frac{2 e^{wu}_w b_{size} \lambda \mu_y}
{p^{idle}_m}+\mu_y^2 k^*}\!-\!\sqrt{\frac{2e^{wu}_w b_{size} \lambda \mu_y}
{p^{idle}_m}}
\end{align}
where $\Delta_{v} b$ refers to the
purely size variation-induced impact on
the {\em fixed-size} buffering scheme.
Examination of $k^*$, at least in principle, 
can elucidate the respective roles of
skewness and coefficient of variation 
in determining the optimal
buffer size $b^*$. The effect of 
size variability could be either
mitigated or augmented by the skewness in size distribution. A
positive skewness alleviates the impact 
of size variability.
In contrast,
a negative skewness strengthens the impact of size variability. In
this case, $k^*$ is positive and grows polynomially with $c_v(y)$,
thereby ensuring $\Delta_{v} b > 0$. 
This requires an additional
buffer size be allocated in order to accommodate the variability in
the data size distribution. 

\begin{figure}[tb]
\centerline{\psfig{file=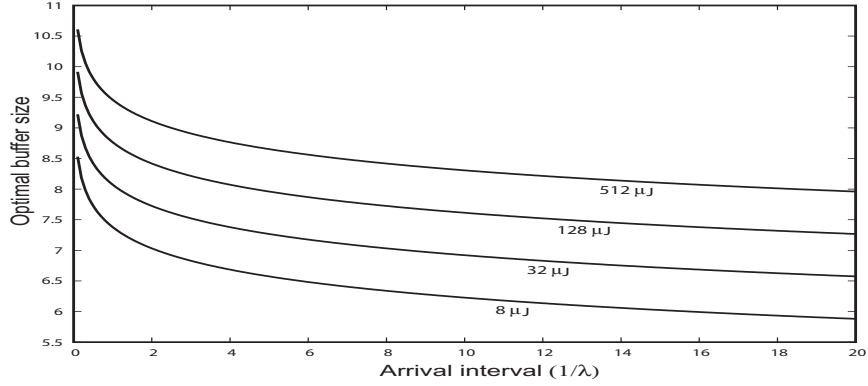,height=2in,width=4.5in}}
\caption{Optimal buffer size $b^*$ vs. 
$1/\lambda$: $\left[p_m^{idle}(0.409\mu W),b_{size}(256 b),\mu_y(256 b),e^{wu}_w (8\mu J,32\mu J,128\mu J, 512\mu J)\right]$}
\label{fig:rate}
\end{figure}

Figure\,\ref{fig:rate} plots the optimal buffer size ($b^*$) as a
function of data arrival interval ($1/\lambda$) when the data
size is exponentially distributed with a mean value of $256$ bytes.
The curves are plotted in a semilog format: the $y$-axis refers to
the optimal buffer size $b^*$ in a {\it log} scale and the
$x$-axis refers to the mean data arrival time $1/\lambda$. 
It shows that an
increase in per radio wakeup energy or in data 
arrival rate (decreasing data arrival interval) 
demands a large buffer size to reduce   
the amortized per radio wakeup cost. This observation agrees with
intuition. 
Combining (\ref{equ:op-fixed-size}) 
and (\ref{equ:fixed-size})
gives the overall power
consumption of the fixed-size buffering as follows:
\begin{align}\label{equ:optimal-size}
\overline{ e^{FS}(b^*)}=
\underbrace{\frac{p_m^{idle} b^*}{b_{size}}\!+
\dfrac{p^{idle}_m \mu_y}{2b_{size}}\!(c_v^2(y)-1)
+\lambda \left(2e^{resyn}_m+e^{wu}_w+\mu_y(e^w_m+e^r_m)\right)}_{\mathrm{buffering}}+
\underbrace{2\lambda\mu_y e^{TX}_w}_{\mathrm{trans/rec}} 
\end{align}

\eqref{equ:optimal-size} yields 
some interesting observations: the
power consumption composition can be roughly 
divided into two
pieces: 1) data transmission/reception power consumption
is linearly proportional to
{\em bandwidth}, {\it i.e.}, $\lambda \mu_y$. 
2) data buffering
power consumption is quite
complicated, hence resists a straightforward
explanation: the buffering power 
consumption not only relies on data arrival
rate $\lambda$ but also depends on the first three moments of the
size distribution. \eqref{equ:optimal-size} shows explicitly
that the data buffering power 
consumption grows asymptotically in proportion to
both $(c_v^2(y)-1)$ and the arrival rate $\lambda$, implying that
the low-variance data size distribution ($c_v(y)\!<\!1$) consumes
less power than the high-variance data size ($c_v(y)\!>\!1$).

\subsection{Fixed-Interval Buffering Scheme}
In this subsection we study the power-aware fixed-interval buffering
scheme, which differs from
its power-aware fixed-size counterpart. The following theorem
gives a direct relation among the 
optimal time interval $T^*$,
the power parameter of radio and memory bank, and data rate and the mean size of the data
distribution. 
\begin{theorem} \label{the:the2}
Let $\lambda$ be a Poisson arrival rate, $\mu_y$ be the mean
size of the data distribution, $e^{wu}_w$ be the per radio wakeup energy, and
$p^{idle}_m$ be the idle state power consumption of a memory bank.
Then, the optimal interval $T^*$ for the {\em fixed-interval} buffering
scheme is:
\begin{align} \label{equ:op-fixed-interval}
T^* &= \sqrt{\frac{2 e^{wu}_w b_{size}}
{p^{idle}_m \lambda \mu_y}}
\end{align}
\end{theorem}

\begin{proof} Let $T$ be the interval of the fixed-interval
buffering scheme. The fixed-interval buffering is a
special case of the renewal process in which the renewal cycle is
constant. Hence the energy consumed in a renewal cycle is expressed as
\begin{align} \label{equ:c1}
e^{FI}(T) & = e^{wu}_w + p^{idle}_m
\sum_{i=1}^{n(T)}
\dfrac{(T-s_i)y_i}
{b_{size}}+
\sum_{i=1}^{n(T)} 
e^{TX}_w y_i+
\sum_{i=1}^{n(T)} 
(e^{RX}_w y_i+ e_w^{wu}) & \\ \nonumber
& +\sum_{i=1}^{n(T)} (e^w_m +e^r_m)y_i+2e^{resyn}_m,
\end{align}
where $n(T)$ is a random variable denoting the number of data
arrivals within the interval $T$, $y_i$ is the size of the {\it i}th
arrived data, and the arrival time $s_i=\sum_{j=1}^i x_j-x_{j-1}$.
The term $\sum_{i=1}^{n(T)} (e^{RX}_w y_i+ e_w^{wu})$
refers to the total reception energy in the interval $T$, which 
involves the energy consumed in receiving the arrived data 
$\sum_{i=1}^{n(T)} e^{RX}_w y_i$, and the energy of 
radio wakeup for data reception $n(T) e_w^{wu}$.
Notice that the radio-triggered power scheme does not 
incur listening power consumption. 
By the renewal reward theory, the long-run mean average energy
consumption is
\begin{align} \label{equ:fixed-interval-buffering}
\overline{e^{FI}(T)} \stackrel{\text{\tiny def}}{=}
\lim_{t\rightarrow \infty} \dfrac{e^{FI}(t)}{t} =\dfrac{
E[e^{FI}(T)]}{T}.
\end{align}
By Wald's equation, the expectation of
the second term in (\ref{equ:c1}) is
\begin{align} \label{equ:equ10}
 \dfrac{p^{idle}_m}{b_{size}}
E \left[\sum^{n(T)}_{i=1}(T-s_i)
I_{\{n(T)>0\}}y_i\right] &
 = \dfrac{p^{idle}_m \mu_y}{b_{size}}
\int^T_0 (T-t) \lambda e^{-\lambda t}
\left (\sum^\infty_{i=1}
\dfrac {(\lambda t)^{i-1}}{(i-1)!}\right ) 
\mathrm{d}t \\ \nonumber
& = \dfrac{p^{idle}_m \mu_y \lambda}{b_{size}} \int ^T_0 (T-t) \mathrm{d}t =
\dfrac{p^{idle}_m \mu_y\lambda T^2}{2b_{size}},
\end{align}
where $I_{\{n(T)>0\}}$ is the indicator function.  
Letting $e_w^{TX}\approx w_w^{RX}$, the expectation of 
the third-sixth terms in (\ref{equ:c1}) are simplified as
\begin{align} \label{equ:ooo}
& E\left[\sum\limits_{i=1}\limits^{n(T)}y_i
(2e^{TX}_w+e^w_m +e^r_m)+2e^{resyn}_m + e^{wu}_w\right] \\ \nonumber
&= E[n(T)] 
E[y_i(2e^{TX}_w+e^w_m +e^r_m)+2e^{resyn}_m +e^{wu}_w]  \\ \nonumber
& = \lambda T 
(\mu_y \left(2e^{TX}_w+e^w_m +e^r_m)+2e^{resyn}_m +w^{wu}_w\right).
\end{align}
Combining (\ref{equ:fixed-interval-buffering})-(\ref{equ:ooo})
gives
\begin{align} \label{equ:qqqq}
\overline{e^{FI}(T)} = \frac{E[\xi_1 (T)]}{T} =
\dfrac{e^{wu}_w}{T} + \dfrac{\lambda T p^{idle}_m \mu_y}
{2b_{size}} + \lambda 
\left(\mu_y (2e^{TX}_w+e^w_m+e^r_m)+2e^{resyn}_m+e^{wu}_w\right)
\end{align}
Taking derivative of (\ref{equ:qqqq}) {\em w.r.t.} $T$ gives
$\dfrac{ \partial \overline{e^{FI}(T)}}{\partial T} =
-\dfrac{e^{wu}_w}{T^2}+ \dfrac{p^{idle}_m
\lambda \mu_y}{2b_{size}}$. 
Resolving
$\frac{\partial \overline{e^{FI}(T^*)}}{\partial T^*} =0$
leads to (\ref{equ:op-fixed-interval}). $\hfill$
\end{proof}

Examination of $T^*$ in (\ref{equ:op-fixed-interval})
reveals the apparent variability immunity of 
the fixed-interval buffering scheme
since $T^*$ only contains the first moment $\mu_y$ of the size distribution.
Substituting (\ref{equ:op-fixed-interval}) into (\ref{equ:qqqq})
gives

\begin{align}\label{equ:optimal-fixed-interval1010}
\overline{e^{FI}(T^*)}&=
\underbrace{\sqrt{\dfrac{2 p^{idle}_m
e^{wu}_w \lambda \mu_y }{b_{size}}}
+\lambda (2e^{resyn}_m+ e^{wu}_w+\mu_y(e^w_m+e^r_m))}_{\mathrm{buffering}} 
 +\underbrace{\lambda 2\mu_y e^{TX}_w }_{\mathrm{trans/rec}} 
\end{align}

In a similar fashion, the power consumption
composition 
in \eqref{equ:optimal-fixed-interval1010} 
also can be divided into
the data transmission/reception and buffering pieces.
The data transmission/reception piece is linearly
proportional to the {\em bandwidth} ($\lambda \mu_y$), while
the data buffer one is
proportional to the square root of the bandwidth
$\sqrt{\lambda \mu_y}$.
Although there is very
little apparent relationship between the fixed-size and the
fixed-interval buffering schemes, both buffering schemes
essentially share the same transmission/reception component but
differ markedly in their data buffering components:
the data buffering component of the fixed-interval 
buffering scheme
is a function of {\em bandwidth}. 
By contrast,
that of the fixed-size buffering scheme is linked to 
the {\em bandwidth} and 
the first three moments of size distribution 
explicitly expressed in term $\mu_y^2k^*$.
\begin{figure}[hbt]
\centerline{\psfig{file=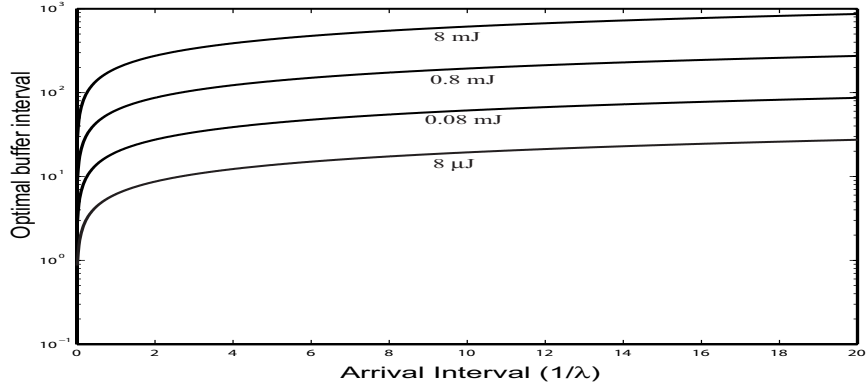,height=2in,width=4.5in}}
\caption{Optimal buffer interval $T^*$ vs. 
$1/\lambda$: $\left[p_m^{idle}(0.409\mu W),b_{size}(256 b),\mu_y(256 b),e^{wu}_w (8\mu J,0.08 mJ,0.8 mJ,8 mJ)\right]$} 
\label{fig:optimalinter}
\end{figure}

The curves in a semilog format in
Figure\,\ref{fig:optimalinter} show that radio wakeup
energy increase results in optimal time interval increase, 
while increasing
idle-mode power consumption in a memory bank reduces the optimal
buffer interval. This can be explained intuitively
as follows: for a high per radio wakeup energy, a large data buffer
(large optimal interval) can effectively 
reduce the amortized per
radio wakeup energy, while a high sleep-mode power consumption would
increase the power consumption of buffering, 
thereby reducing optimal buffer interval $T^*$.

Let us digress a little bit from the  main 
derivation to examine
the no-buffer scheme: a special case of the {\it fixed-interval}
buffering scheme in which the sensor node transmits data immediately
upon receipt of measured data. Mathematically, this corresponds to a
case where the mean buffer interval $T\!=\!\frac{1}{\lambda}$.
The following corollary deals with the no-buffer scheme. 

\begin{corollary} \label{coro:nobuffer}
The long-run mean average energy consumption of the no-buffer
scheme, denoted by $\overline{e(nb)}$, is
\begin{align}  \label{equ:nobuffer}
\overline{e(nb)}
=\lambda (2e^{wu}_w\!+\mu_y(2e^{TX}_w+e^w_m+e^r_m)
+2e^{resyn}_m )
\end{align}
\end{corollary}

\begin{proof} Consider a renewal reward process with the
cycle length ($\mathrm{L_c}$)being equal to the data arrival interval $T=1/\lambda$. 
Thus the energy consumed in this cycle is
\begin{align} \label{equ:no_buffer}
&e(nb)=e^{wu}_w+(2e^{TX}_w+e^w_m+e^r_m)y_i+2e^{resyn}_m +e^{wu}_w,
\end{align}
where $y_i$ is the size of {\it i}th arrived data.
(\ref{equ:no_buffer}) is simply attained by removing the buffering
factors (terms) in (\ref{equ:equ10}). Using the same argument in
proving Theorem~\ref{the:the2} we get
$\overline{e(nb)}=\frac{E[e(nb)]}{\frac{1}{\lambda}}$ in
(\ref{equ:nobuffer}) since $E[L_c]\!=\!\frac{1}{\lambda}$. $\hfill$
\end{proof}

\begin{figure}[tbh]
\centerline{\psfig{file=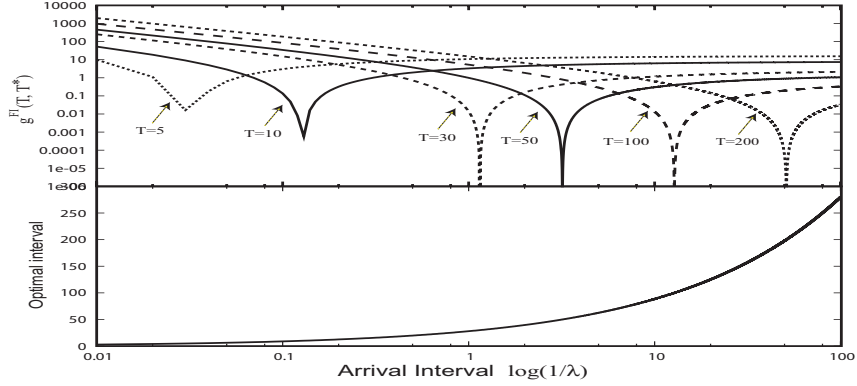,height=2in,width=4.5in}}
\caption{Top: $g^{FI}(T,T^*)$ vs.
arrival interval $\log(1/\lambda$). Bottom: $T^*$ vs. arrival interval $\log(1/\lambda)$: $\left[p_m^{idle}(0.409\mu W),e^{wu}_m(80 \mu J),b_{size}(128b),
\mu_y(64b)\right]$}
\label{fig:interval-compare}
\end{figure}

Define a function $g^{FI}(T,T^*)$ to 
quantify the differential gain of
the optimal power-aware fixed-interval buffering over a power-oblivious 
buffering scheme.
\begin{align} \label{equ:general-interval}
g^{FI}(T,T^*)= \overline{e^{FI}(T)}-
\overline{e^{FI}(T^*)}=\frac{e^{wu}_w}{T} +\dfrac{T p^{idle}_m\lambda \mu_y}
{2 b_{size}} -
\sqrt{\frac{2 p^{idle}_m e^{wu}_w \lambda \mu_y}{b_{size}}},
\end{align}
where $T^*$ is the optimal interval and $T$ is chosen arbitrarily. 

Results of $g^{FI}(T,T^*)$ are plotted in
Figure\,\ref{fig:interval-compare}. The main trend is that the
optimal buffer interval $T^*$ grows as the square root of
$\frac{1}{\lambda}$ (see bottom graph), and that a dip in each curve
occurs when arbitrarily chosen $T$ happens to be 
in the vicinity of $T^*$. The
differential gain $g^{FI}(T,T^*)$ arises sharply
when $T$ deviates from the optimal buffer interval $T^*$. 
This implies that a blind selection of buffer interval $T$
is very likely to incur an excessive energy consumption.

\section{Performance Comparison}
In this section we attempt to answer two fundamental questions: 1)
how much power saving via power-aware buffering can be achieved in
comparison to the no-buffer scheme\,? 2) 
the {\it fixed-size} buffering or its {\it fixed-interval}
counterpart, which one performs better ?

\subsection{Comparison between the no-buffer scheme and
power-aware buffering schemes}

It is evident that the {\it
fixed-interval} buffering always 
outperforms the no-buffer scheme as
the latter is a special case 
of the former. Below we compare the no-buffer 
scheme with the {\it fixed-size} buffer one.

The differential power consumption between the no-buffer 
and
optimal fixed-size buffering schemes is expressed as
\begin{align} \label{equ:inequality}
\overline{e(nb)}- \overline{e^{FS}(b^*)} = \lambda e^{wu}_w - \frac{p_m^{idle}}{b_{size}} \left (b^* +
\dfrac{\mu_y (c^2_v(y)-1)}{2} \right ) >0
\end{align}

\eqref{equ:inequality} does in fact constitute an {\it incentive condition}
under which the optimal {\it fixed-size} buffering scheme outperforms
the no-buffer scheme in power conservation. 
It shows that increasing
variability $c_v(y)$ in effect erodes
the gain brought out by the power-aware {\it fixed-size} buffering
scheme, hence shrinks the incentive area. A positive skewness in size
distribution can neutralize, to some extent, the size
variability-induced impact. 
While in general this incentive
condition could be profoundly affected by various intertwined and
correlated factors, we explicitly derive closed-form expressions under some 
restricted scenarios:

1) Exponential data size distribution $y$ with a mean of $\mu_y$.
Under this condition, $k^*$ is reduced to zero according to
(\ref{equ:simplekey}), the incentive condition thus becomes
\begin{align} \label{equ:exponential}
\overline{e(nb)}- \overline{e^{FS}(b^*)}  = \lambda e^{wu}_w -
\frac{p_m^{idle} b^*}{b_{size}} =  \lambda
e^{wu}_w - \sqrt{\frac{2 \lambda p_m^{idle} e_w^{wu}
\mu_y}{b_{size}}} >0
\end{align}

It is obvious that (\ref{equ:exponential}) holds as long 
as $\lambda e^{wu}_w\!>\!\frac{2 p^{idle}_m
\mu_y}{b_size}$ is met. This incentive condition can be rewritten in
a structurally meaningful form that emphasizes the distinction
between hardware parameters and operational 
requirement as follow
\begin{align} \label{equ:cohesive}
\dfrac{e^{wu}_w}{\frac{p^{idle}_m}{b_{size}}} > \frac{2
\mu_y}{\lambda}.
\end{align}
Using {\it byte-second} as a quantifiable unit, the left-hand side
of (\ref{equ:cohesive}) is related to hardware power parameters:
the ratio of radio wakeup energy to the per-byte idle-mode power
consumption of a memory bank. The right-hand side, on the other hand,
is related to the operational requirement: the ratio of the mean
data size to the data arrival rate. 
For given power parameters and $\mu_y$, there exists a critical value for
$\lambda_c\!=\!\frac{2\mu_y p^{idle}_m}{b_{size} e^{wu}_w}$.
When $\lambda>\lambda_c$, the {\it fixed-size} buffering scheme is
preferred. Otherwise, the no-buffer scheme is preferred.
A high ratio of the per radio wakeup energy to the 
per-byte
idle-mode memory power consumption favors a 
large buffer size. On the
other hand, the benefit of data buffering is
diminished as $\frac{\mu_y}{\lambda}$ increases.

2) Erlang size distribution $y$ with parameters
$(\alpha,\lambda_{\alpha})$. This corresponds to the case 
in which $c_v(y)=\frac{1}{\sqrt{\alpha}},
\gamma(y)=\frac{2}{\sqrt{\alpha}},
\mu_y=\dfrac{\alpha}{\lambda_{\alpha}},
\sigma_y=\dfrac{\sqrt{\alpha}}{\lambda_{\alpha}}$, then
$k^*=\!\frac{1}{12}(1-\frac{1}{\alpha^2})$.
The incentive condition is thus expressed as

\begin{align} \label{equ:erlang}
\overline{e(nb)}- \overline{e^{FS}(b^*)}  = \lambda e^{wu}_w
 -\frac{p_m^{idle}}{b_{size}}
\left(b^*-\frac{\mu_y}{2} (1-\frac{1}{\alpha})\right)  > 0
\end{align}
where
\begin{align} \label{equ:kk}
b^*= \sqrt{\frac{2 \lambda b_{size} e_w^{wu} \mu_y}{p^{idle}_m}
+ \frac{\mu_y^2}{12}(1-\frac{1}{\alpha^2})}
\end{align}
(\ref{equ:kk}) shows that, as compared with an
exponential size distribution ($\alpha=1$), an additional buffer size
needs to be allocated when the shape parameter $\alpha>1$. To study
the impact of data arrival rate $\lambda$, we define $f(\lambda)=
\overline{e(nb)}\!-\!\overline{e^{FS}(b^*)}$. Differentiating
$f(\lambda)$ and solving $f^\prime(\lambda^*)=0$ gives
\begin{align} \label{equ:minim}
\lambda^* =\frac{\mu_y p^{idle}_m}{24 b_{size} e^{wu}_w} \left
(11 +\frac{1}{\alpha^2} \right)
\end{align}

\noindent Observe that $f(\lambda)$ has a global minimum point at
$\lambda^*$ since $f^{\prime\prime}(\lambda^*)>0$. This implies
that $f(\lambda^*) < f(\lambda), \lambda \in {\cal R}^+$ and
$\lambda \not=\lambda^*$. Substituting (\ref{equ:minim}) into
(\ref{equ:erlang}) gives
\begin{align} \label{equ:llll}
f(\lambda^*)= -\frac{\mu_y p^{idle}_m}{24 b_{size}} \left (1
+\frac{12}{\alpha}- \frac{1}{\alpha^2} \right) < 0, \ \  \alpha \geq
1
\end{align}
Since $f(0)= \frac{p^{idle}_m\mu_y}{2 b_{size}}
\left(1-\frac{1}{\alpha}-
\frac{1}{\sqrt{3}}\sqrt{1-\frac{1}{\alpha^2}}\right)\leq 0, \alpha
\geq 1$, and $f(\lambda) \rightarrow e^{wu}_w \lambda > 0$ when
$\lambda$ is sufficiently large, one concludes 
that there exists a critical data rate
$\lambda_c > \lambda^*$ such that the no-buffer scheme outperforms
the {\it fixed-size} buffering scheme 
when $\lambda \leq \lambda_c$ and the {\it fixed-size} buffering scheme
is preferred when $\lambda>\lambda_c$. Figure\,\ref{fig:optimalsize}
plots $f(\lambda)$ as a function of $1/\lambda$ with different shape
parameters in a semilog format. 
Figure\,\ref{fig:optimalsize} illustrates that
increasing shape parameter $\alpha$ (less variation) results in a
decreased critical data rate $\lambda_c$ (increasing $T_c= 1/\lambda_c$),
indicating that the smaller the data size variation,
the larger the incentive region. 
It is worth noting that there exists an inherent 
tradeoff between buffering and responsiveness:
the no-buffer scheme achieves real-time responsiveness
at the expense of power consumption, while 
the power-aware buffering to some extent can save
power consumption, but at the price of reduced responsiveness.  

\begin{figure}[tb]
\centerline{\psfig{file=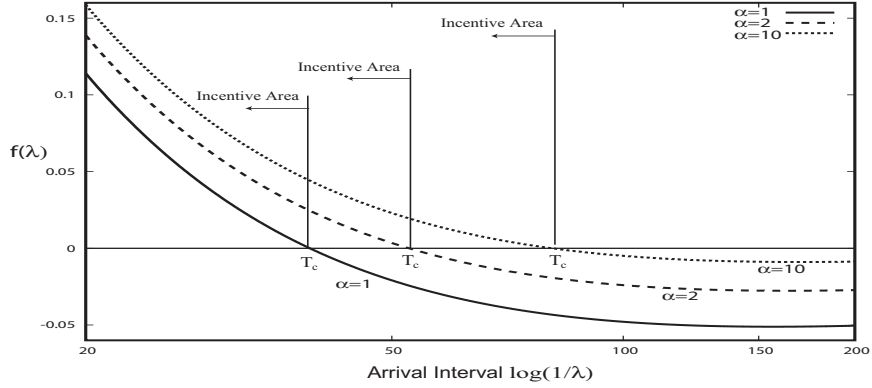,height=2in,width=4.5in}}
\caption{Incentive area vs. arrival interval:
$\left[p_m^{idle}(0.409\mu W),b_{size}(256b),e^{wu}_m(80\mu J),\mu_y(64b)\right]$} 
\label{fig:optimalsize}
\end{figure}

\subsection{Comparison between Fixed-size and 
Fixed-interval Buffering Schemes}

One question arises naturally: which buffering scheme is more power
efficient? the power-aware {\em fixed-size} scheme or the
power-aware {\em fixed-interval} one. While in general there is no
simple answer to this question, there is a definite answer under some
special circumstances. We begin with an easy lemma as below. 

\begin{lemma} \label{lem:compare}
Let $f(x)\!=\!\sqrt{ax}-\!\sqrt{ax+b}\!+\!c, x\geq 0$,
where $a,b$, and $c$ are positive, and
$\sqrt{b}<c$. Then $f(x)\!>\!0$ for $x\!\geq\!0$. 
\end{lemma}
\begin{proof}
Take derivation of $f(x)$, we obtain
\begin{align} \label{equ:moni}
f^\prime(x)= \dfrac{a(\sqrt{ax+b}-\sqrt{ax})}{2\sqrt{ax(ax+b)}}>0, x
> 0
\end{align}
This means that $f(x)$ is monotonically increasing for $x\!\geq\!0$.
Then $f(0)\!=\!\min\limits_{x \in
(0,\infty)}\!f(x)$. It can be inferred that $f(x)>f(0)>0$ for $x\geq 0$ since 
$f(0)\!=\!-\sqrt{b}+c>0$. $\hfill$
\end{proof}

\begin{theorem} \label{theo:compare}
The power-aware {\em fixed-size} buffering scheme
is more power-efficient than
the {\em fixed-interval} counterpart
when the data size is constant, while both 
the buffering
schemes perform equally well 
when the data size is exponentially distributed.  
\end{theorem}
\begin{proof}
Define $g(T^*,b^*)$ to denote
the power consumption differential between the
{\it fixed-interval} and
{\it fixed-size} buffering schemes
as follows
\begin{align} \label{equ:compare}
g(T^*,b^*) & = \overline{e^{FI}(T^*)}
- \overline{e^{FS}(b^*)}  =\sqrt{\dfrac{2 \mu_y \lambda p^{idle}_m
e^{wu}_w }{b_{size}}} -\frac{p_m^{idle}}{b_{size}} \left(b^*
+ \dfrac{\mu_y(c^2_v(y)-1)}{2}
\right)
\end{align}

\noindent (I) Constant data size: Since
$\sigma_y\!=\!0$ and $k^*\!=\!1/12$, thus
(\ref{equ:compare}) becomes
\begin{align} \label{equ:compare1}
g(T^*,b^*) = \sqrt{\dfrac{2 \mu_y \lambda p^{idle}_m
e^{wu}_w }{b_{size}}} -\sqrt{\dfrac{2 \mu_y \lambda p^{idle}_m
e^{wu}_w }{b_{size}} +\dfrac{(\mu_y p^{idle}_m)^2}{12 b^2_{size}}}
+\dfrac{p^{idle}_m \mu_y}{2b_{size}}
\end{align}
Let $a\!=\!\frac{2\mu_y p^{idle}_m e^{wu}_w}{b_{size}},\,b=\!\frac{(\mu_y p^{idle}_m)^2}{12 b_{size}^2},
\,c\!=\!\frac{p^{idle}_m \mu_y}{2 b_{size}}$.
Based on Lemma~\ref{lem:compare}, we have
$g(T^*,b^*)> 0$,
{\it i.e.}, $\overline{e^{FI}(T^*)}> \overline{e^{FS}(b^*)}$. \\

\noindent (II) Exponential size distribution:
Since $\sigma_y\!=\!\mu_y,c_v(y)\!=\!1,\gamma(y)\!=\!2$, thus
$k^*\!=\!0$, (\ref{equ:compare}) becomes
\begin{align} \label{equ:compare2}
g(T^*,b^*)\!=\!\sqrt{\dfrac{2 \mu_y \lambda p^{idle}_m
e^{wu}_w}{b_{size}}}\!-\!\sqrt{\dfrac{2 \mu_y \lambda p^{idle}_m
e^{wu}_w }{b_{size}}}=0.
\end{align}
Combining (I)-(II) completes the proof. $\hfill$
\end{proof} 

Consider a hyper-exponential size
distribution: $f_h(y)=\sum_{i=1}^2 \frac{1}{\mu_i}
p_i \exp{(-y/\mu_i)}$ and $\sum_{i=1}^2 p_i\!=\!1$. Letting
$p_1\!=\!p,p_2\!=\!1-p$. 
Since there is no closed-form
expression for $g(T^*,b^*)$, thus
a numerical method is used 
to compute different values for $c_v(y)$ and
$\gamma(y)$. This is achieved by varying the value of $p$
while maintaining $\mu_y$
constant. Figure\,\ref{fig:size-interval} presents $g(T^*,b^*)$ as a
function of $\frac{1}{\lambda}$ with different $c_v(y)$ and
$\gamma(y)$. It shows that $c_v(y)$ has a substantial impact on
$g(T^*,b^*)$, whereas $\gamma(y)$ plays a marginal
role. For example, when $c_v(y)\!=\!1.72$ and $\gamma(y)=2.72$,
$g(T^*,b^*)\!=\!\overline{e^{FI}(T^*)}\!-\!\overline{e^{FS}(b^*)}\!<\!0$. 

\begin{figure}[htb]
\centerline{\psfig{file=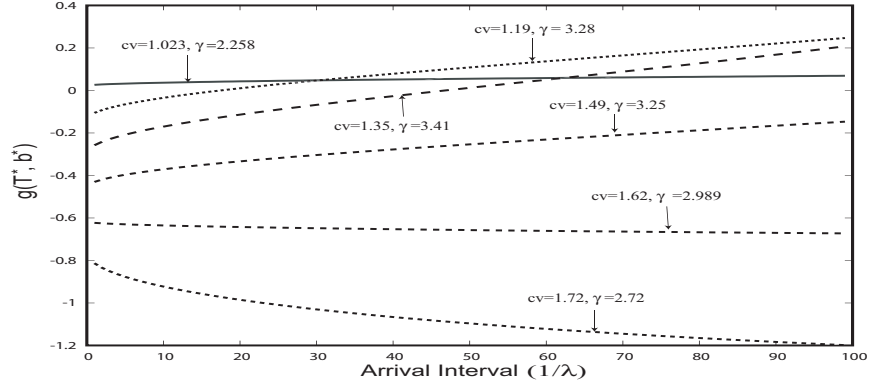,height=2in,width=4.5in}}
\caption{$g(T^*,b^*)$ vs. $1/\lambda$:
$\left[p_m^{idle}(0.409\mu W),e^{wu}_m(80\mu J),b_{size}(256b)\right]$}
\label{fig:size-interval}
\end{figure}

Combining the above 
analysis and numerical calculation leads to the
conclusion that 
the size variation-induced effect is a non-negligible
role in determining the relative advantages of the {\it
fixed-interval} and {\it fixed-size} buffering schemes: when the size
distribution is of low-variance, the {\it fixed-size} buffering scheme
outperforms the {\it fixed-interval} counterpart. When the
size distribution is of high-variance, the {\it fixed-interval}
buffering scheme is more energy-efficient than the {\it
fixed-size} one. Between these two extremes in size variability,
the relative power efficiency of these two 
buffering schemes 
depends on data arrival rate $\lambda$.

To illustrate the efficacy of the power-aware
fixed-size buffering scheme over a power-oblivious one,
we define a function $g^{FS}(b,b^*)$ as follows
\begin{align} \label{equ:general}
 g^{FS}(b,b^*) = \overline{e^{FS}(b)}-
\overline{e^{FS}(b^*)} =(b^* -b)\left(\frac{\lambda e^{wu}_w \mu_y +\dfrac{p^{idle}_m \mu_y^2 k^*}
{2 b_{size}}}{b \cdot b^*} -
\frac{p^{idle}_m}{2 b_{size}}\right ),
\end{align}
where $b^*$ denotes the optimal buffer size (see (\ref{equ:op-fixed-size})), 
and $b$ is  arbitrarily chosen ($b \in {\cal R}$ and $ b\not= b^*$). 
In practical terms, the amount of power consumed 
in the idle-mode memory banks is 
quantized into discrete levels as $
p_m^{idle}\!\left \lceil \frac{b}{b_{size}} \right \rceil$, 
where $\lceil \rceil$ is the ceiling function. 
This implies that only certain discrete 
power states are allowed. For example, if $b_{size}\!=\!128$,
then two memory banks will be used when the
buffer size $b$ falls within the range of $[128,255\,]$.

To visualize this quantization impact, 
we examine the function 
$g^{FS}(b,b^*)$ when $k^*=0$ (an exponential
size distribution). Figures\,\ref{fig:256}-\ref{fig:1280} plot
$g^{FS}(b,b^*)$ as a function of $1/\lambda$ under different values
of $\mu$ (mean data size). The y-axis in the top graph in
Figure\,\ref{fig:256} represents $g^{FS}(256,b^*)$ in the
unit of ($\mu W$). The y-axis in the bottom graph in
Figure\,\ref{fig:256} denotes the optimal buffer size $b^*$ in a
multiple of $b_{size}$. The x-axis
refers to the mean arrival interval $1/\lambda$ in a $\log$
scale.

With increases in $\frac{1}{\lambda}$, the optimal buffer size $b^*$
decreases. The curve of $g^{FS}(256,b^*)$ gradually declines as
$\frac{1}{\lambda}$ increases, as shown in Figure\,\ref{fig:256}. In
contrast, in Figure\,\ref{fig:1280} with increases in
$\frac{1}{\lambda}$, the differential gain $g^{FS}(1280,b^*)$
initially monotonically decreases and then increases after reaching
the lowest point at which the buffer size $b$ is optimal at
$10b_{size}=1280$. Observe that the
memory-bank-size quantized effect produces a stair-like relation
between the amount of power consumed and $\frac{1}{\lambda}$. For
example, the optimal buffer size is $6b_{size}$ 
when $\frac{1}{\lambda} \in [5.435,7.825)$. The larger the
arrival interval $\frac{1}{\lambda}$, the more pronounced (a
wider stair space) the quantized effect.

\begin{figure}[ht]
\centerline{\psfig{file=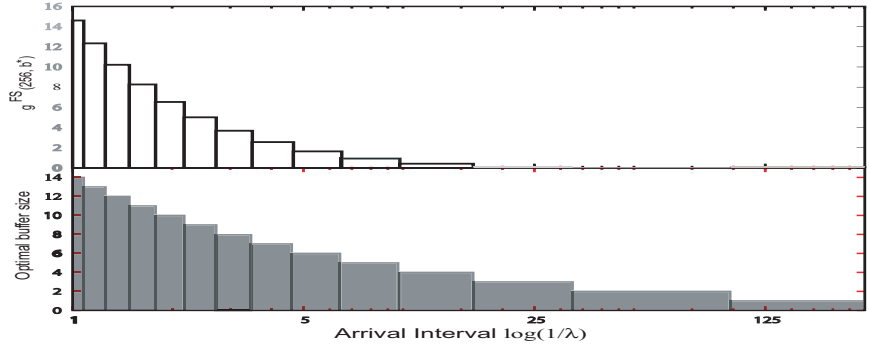,height=2in,width=4.5in}}
\caption{Top: $g^{FS}(256,b^*)$ vs arrival rate $\log(1/\lambda)$
Bottom: optimal buffer size $b^*$ vs. $\log(1/\lambda)$
$\left[b_{size}(128b),\mu_y(64b),e^{wu}_w(80\mu J),p^{idle}_m=0.409 \mu W\right]$} 
\label{fig:256}
\end{figure}
\begin{figure}[hbt]
\centerline{\psfig{file=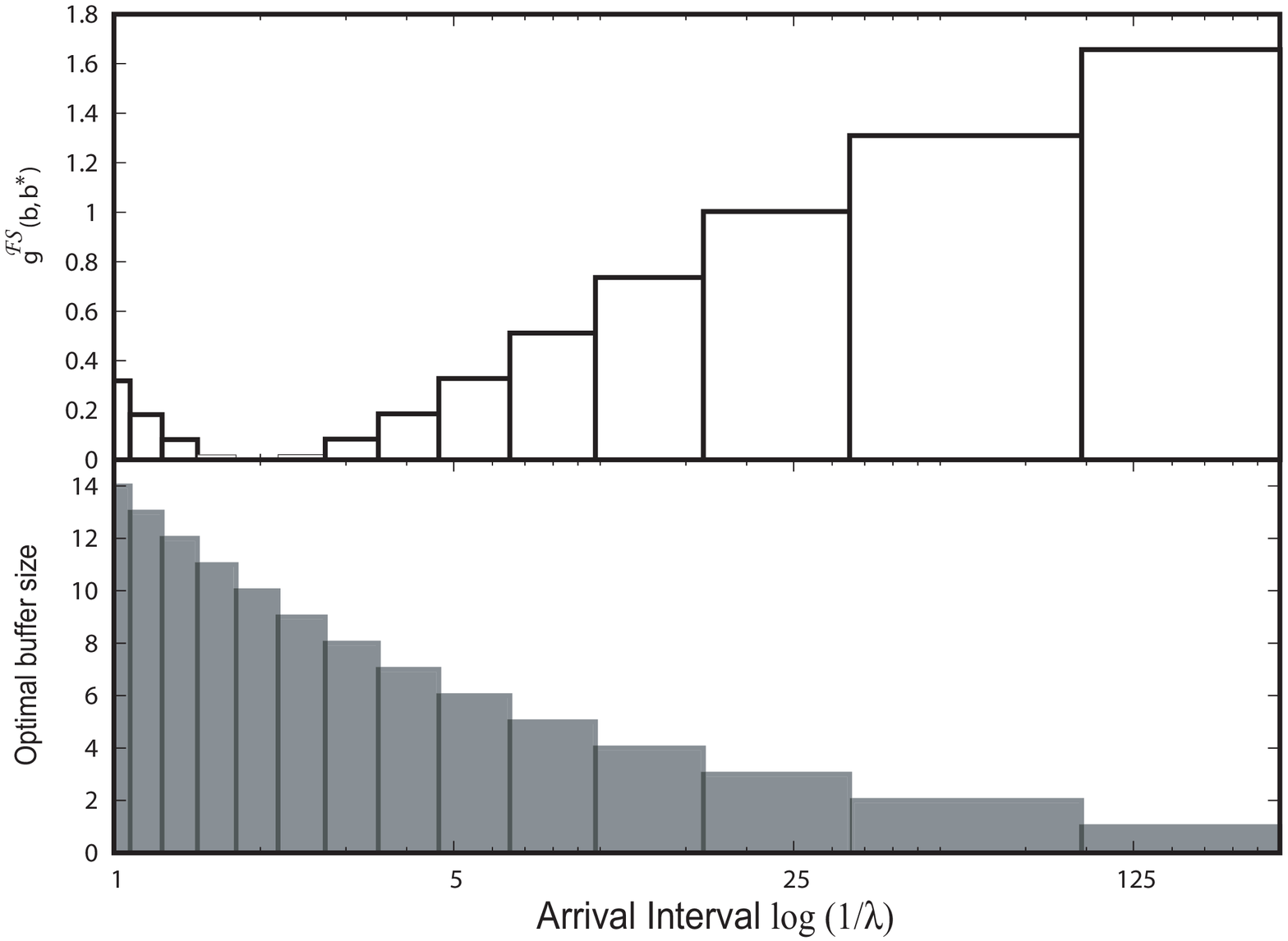,height=2in,width=4.5in}}
\caption{Top: $g^{FS}(1280,b^*)$ vs. $\log(1/\lambda)$.
Bottom: optimal buffer size $b^*$ vs. $\log(1/\lambda)$.
$\left[b_{size}(128b),\mu_y(64b),e^{wu}_w(80\mu J),
p^{idle}_m(0.409\mu W)\right]$} 
\label{fig:1280}
\end{figure}

\begin{figure}[hbt]
\centerline{\psfig{file=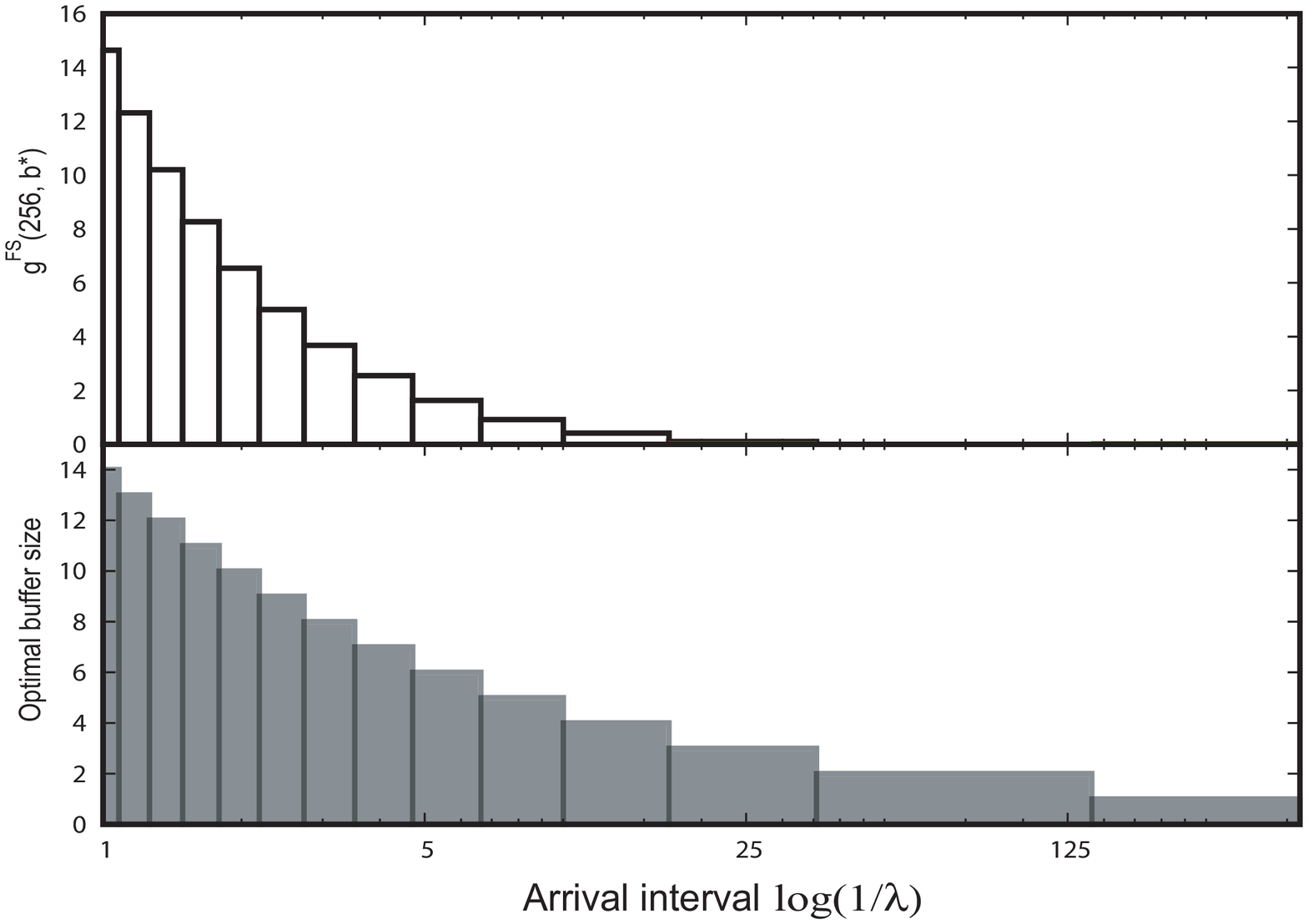,height=2in,width=4.5in}}
\caption{Top: $g^{FS}(256,b^*)$ vs. $\log(1/\lambda)$)
Bottom: optimal buffer size $b^*$ vs. $\log(1/\lambda)$ 
$\left[b_{size}(128b),\mu_y(64b),e^{wu}_w(80\mu J),p^{idle}_m(0.409\mu W),
c_v(y)(1.723),\gamma(y)(2.718)\right]$}
\label{fig:256v}
\end{figure}
\begin{figure}[tbh]
\centerline{\psfig{file=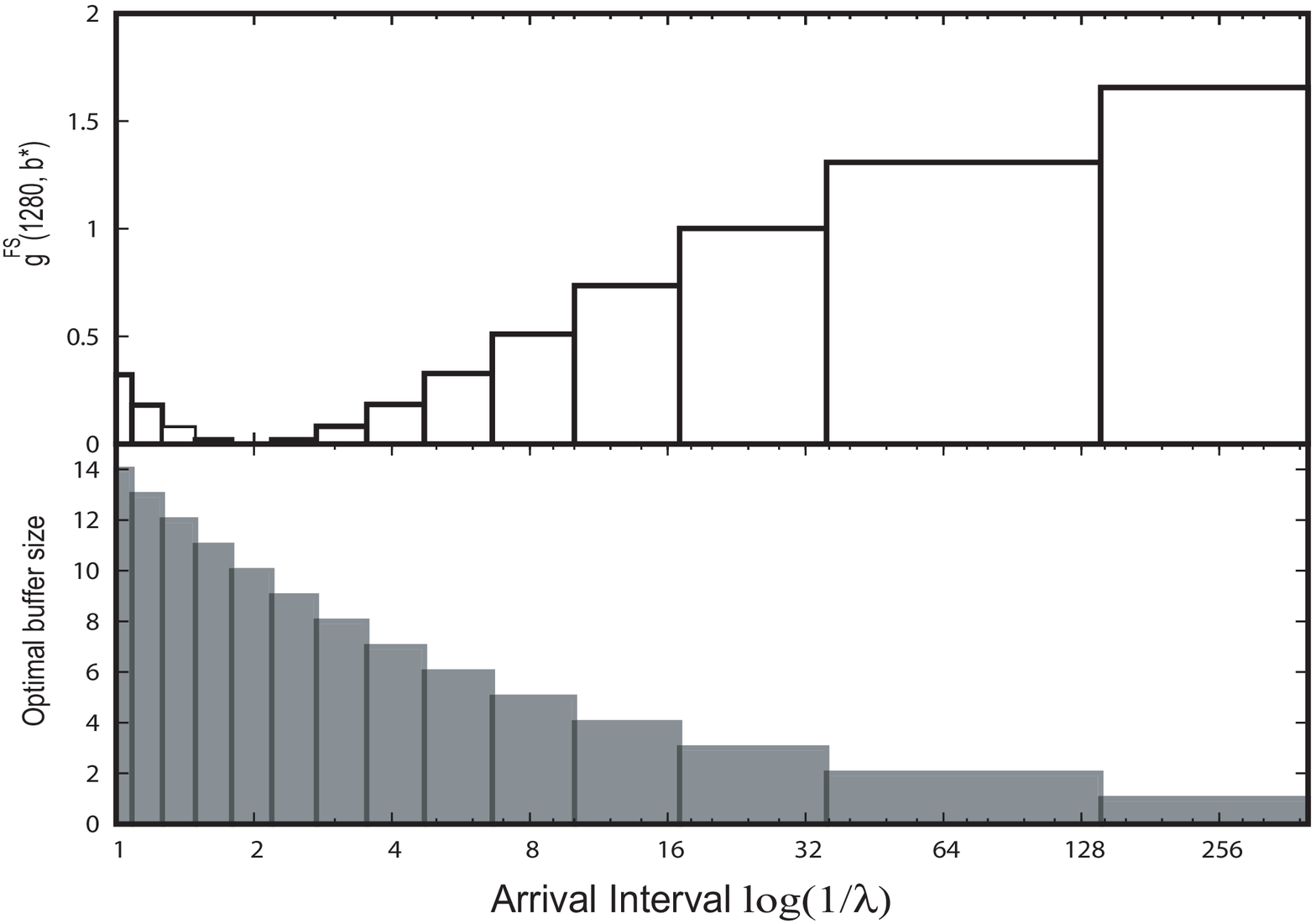,height=2in,width=4.5in}}
\caption{Top: $g^{FS}(1280,b^*)$ vs. $\log(1/\lambda)$
Bottom: optimal buffer size $b^*$ vs. $\log(1/\lambda)$ 
$\left[b_{size}(128b),\mu_y(64b),e^{wu}_w(80\mu J),p^{idle}_m(0.409\mu W),c_v(y)(1.723),\gamma(y)(2.718)\right]$}
\label{fig:1280v}
\end{figure}

To evaluate the size variation-induced effects, 
we calculated $g^{FS}(b,b^*)$ under a
hyperexponentially distributed data size in which $c_v(y)=1.732$
and $\gamma(y)=2.718$ and plot both $g^{FS}(256,b^*)$ and
$g^{FS}(1280,b^*)$ as a function of $\frac{1}{\lambda}$ in
Figures~\ref{fig:256v}-\ref{fig:1280v}, respectively. In comparison
to the absence of size variation as shown in
Figures\,\ref{fig:256}-\ref{fig:1280}, 
the size variation effect
becomes inconsequential when $\lambda$ is high, 
but has a substantial impact when $\lambda$ is low (see Figure\,\ref{fig:tt}).

\begin{figure}[hbt]
\centerline{\psfig{file=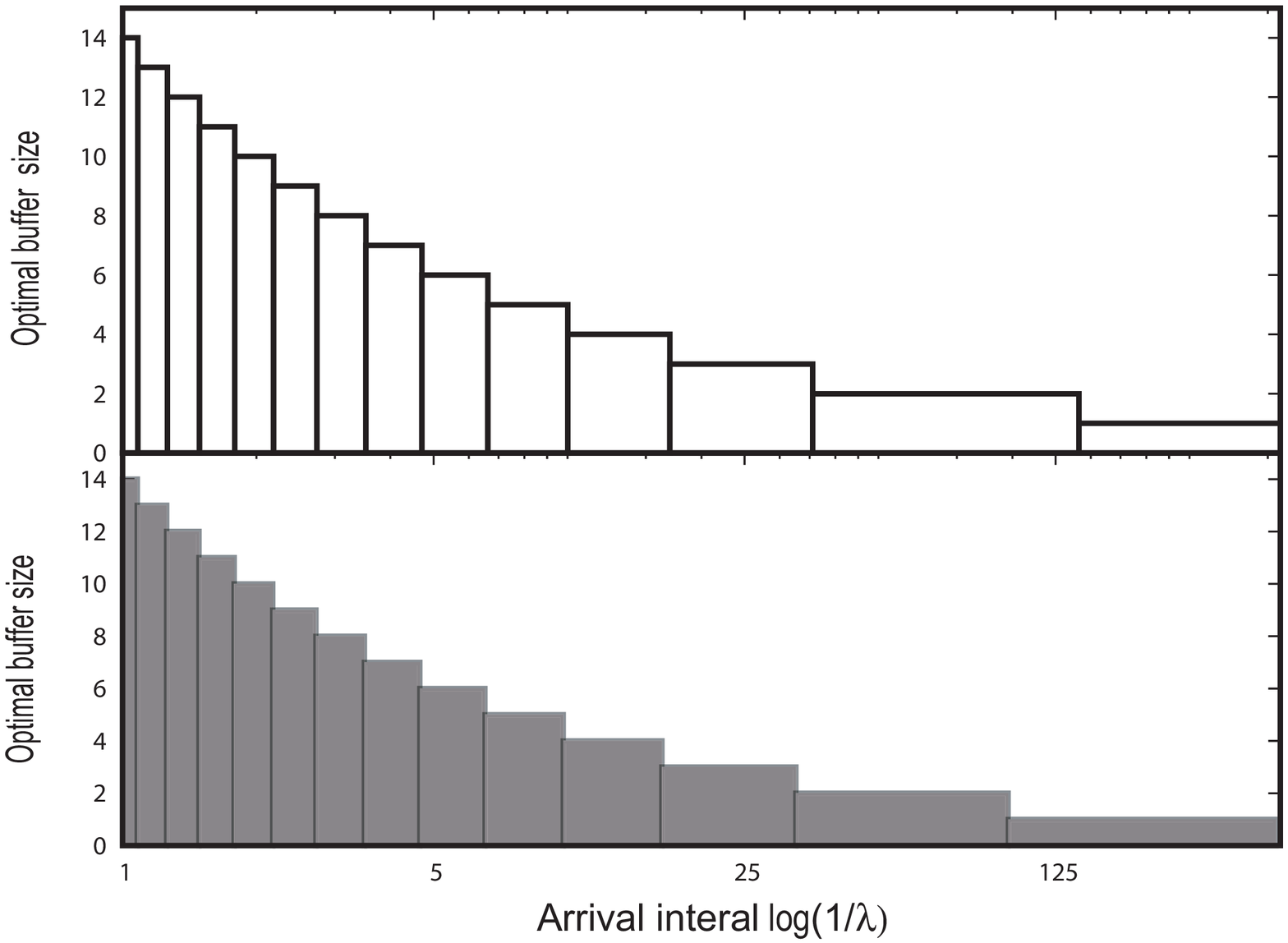,height=2in,width=4.5in}}
\caption{Top: Optimal buffer size $b^*$ vs. 
$\log(\frac{1}{\lambda}):(c_v(y)(1.723),\gamma(y)(2.718))$.
Bottom: Optimal buffer size $b^*$ vs. 
$\log(\frac{1}{\lambda}):(c_v(y)(0))$
$[b_{size}(128b),\mu_y(64b),e^{wu}_w[80\mu J),p^{idle}_m(0.409\mu W)]$} 
\label{fig:tt}
\end{figure}
This phenomenon can be explained by reference to (\ref{equ:tttt5}). 
$\Delta_{v} b$ in (\ref{equ:tttt5}) shows that 
the size variation effect is proportional to
the term $\mu_y^2 k^*$. 
Hence the relative effect of size
variation on the optimal 
buffer size $b^*$ can be expressed as
\begin{align} \label{equ:relativeb}
\frac{\Delta_{v} b}
{\sqrt{\frac{2 \lambda e^{wu} w b_{size} \mu_y}
{p^{idle}_m}}}
= \sqrt{1+ \frac{\mu_y p^{idle}_m k^*}{2 \lambda e^{wu}_m b_{size}}}-1
\end{align}

(\ref{equ:relativeb}) indicates that for a 
given size distribution
(given $c_v(y)$ and $\gamma(y)$, thereby $k^*$), 
the size
variability effect
becomes prominent for low duty-cycle sensor nodes as
(\ref{equ:relativeb}) increases with decreasing $\lambda$ (low duty cycle). 

\section{Effect of power-aware buffering schemes on life span}

In this section we will use a concrete example to quantify
the benefits of power-aware schemes in terms of the 
lifespan extension.

\begin{table}[htb]
\centering
\begin{tabular}{|c|c|c|c|c|} \hline
$e^{wu}_w$ & $p^{idle}_m$ & $e^{TX}_w$ &
$e^{r}_m+e^{w}_m$ & $e^{resyn}_m$ \\ \hline 
$80~\mu J $ & $0.409~\mu W$ & $8.976~\mu J /byte$
& $36\cdot  10^{-3} \mu J /byte $ & $0.912~\mu J$  \\ \hline
\end{tabular}
\caption{Power parameters of radio and memory bank}
\label{tab:tab10}
\end{table}

Power parameters in Table~\ref{tab:tab10} are either directly
obtained or indirectly derived from 
the literature. 
In CC2420-802.15.4 radio specification 
[\citeNP{CC2420}], the transmission power 
is $-25\,dBm$ with a data rate of $250$ kbps, and the current draw is
$8.5\,mA\,(3.3V)$. The $250$ kbps 
is the optimal rate in an ideal environment, which 
may not make any practical sense.  
In practical terms, the rate is assumed to be
$25$ kbps. Consequently, 
the one-byte transmission 
energy is calculated as $e^{TX}_w=\frac{8.5\times 3.3\times
8}{25}=8.976 \mu J/byte$. The power consumption for an idle-mode SRAM memory bank is
$p^{idle}_m=0.409 \mu W$ [\citeNP{Hempstead2005}]. 
Due to unavailability of actual power data of SRAM bank in the literature, the  
$e^r_m, e^w_m, e^{resyn}_m$ are approximated 
by using Rambus DRAM power data. A read or write operation on
Rambus DRAM takes $60 ns$ and consumes $300 mW$,
{\it i.e.}, $e^r_m=e^w_m=18\cdot 10^{-3} \mu J/byte$. A
transition from the powerdown mode of a Rambus DRAM to the active
mode consumes $152 mW$ and takes $6000 ns$ [\citeNP{Fan2001}]. Thus
a resynchronization energy is
$e^{resyn}_m =0.912\,\mu J$.
A radio wakeup energy is assumed to be $e^{wu}_w=80\,\mu J$
as it is not found in literature.
Let the supply voltage be $3.3\,V$, and
the lifespan of two AA batteries be $2700\,mA h$ 
[\citeNP{Levis2005}].

Consider a scenario with a Poisson arrival and 
constant data size (implying $k^*=1/12$).
Letting $y_i=\mu_y=b_{size}=128b$.
Based on  Table~\ref{tab:tab10}, 
the comparison results between 
the optimal fixed-size buffering, 
optimal fixed-interval buffering,  
and an power-oblivious buffering with buffer size
of $256$ are tabulated in Table~\ref{tab:tabcomp}.
 
With $\lambda=0.5$, by 
(\ref{equ:op-fixed-size}),
the optimal buffer size is calculated as
$b^*=1790.54$(b),
thus 
the number of memory banks involved is $\lceil
\frac{b^*}{b_{size}}\rceil\!=14$. 
By (\ref{equ:optimal-size}),
we obtain $\overline{e^{FS}(b^*)}=1197.660858\,\mu W$.
The amount of current draw is
$\frac{1197.660858}{3.3}=362.927533 \mu A$,
and the lifespan is 
$\frac{3.3*2700000}{1197.660858*24*365}=0.849258$(yr). 
Similarly, 
the power consumption of the power-oblivous 
and optimal fixed-interval schemes
are calculated as 
$\overline{e^{FS}(256)}=1212.357021\mu W$ and 
$\overline{e^{FI}(T^*)}=1197.864140\,\mu W$.
This implies that 
the power-aware fixed-size scheme
outlives the power-oblivious one by $3.75761$ days 
and outlives the optimal fixed-interval one
by $1.262509$ hours.  

Table~\ref{tab:tabcomp} shows that 
the role of the power-aware buffering
becomes more 
prominent when the node operates at a low duty cyle. 
For example, with $\lambda =1$,
the optimal fixed-size buffer scheme 
outlives the power-oblivious one by $2.069278$ days,
when $\lambda=0.1$,  
the optimal fixed-size buffer scheme 
outlives the power-oblivious one by $11.797951$ days.
A side-by-side comparison in Table~\ref{tab:tabcomp} suggests that 
the fixed-size buffering scheme performs slightly better
than the fixed-interval counterpart. 
However, this marginal advantage will be  
disappeared in the presence of data-size 
variations. 

\begin{table}[htb]
\centering
\begin{tabular}{|c||c|c||c|c||c|c||} \hline
$\lambda$ & $\overline{e^{FS}(b^*)}$ & 
lifespan &
$\overline{e^{FI}(T^*)}$ & lifespan & 
 $\overline{e^{FS}(256)}$ 
& lifespan  \\
$(1/s)$ & $(\mu W)$ & (yr) &
$(\mu W)$ & (yr) &  $(\mu W)$
& (yr) \\ \hline
$1$ & 
$2392.1739$ &
$0.4252$ & $2392.3775$ & 
$0.4252$ & $2424.5010$ & $0.4195$ \\ \hline
$0.9$ & 
$2153.3299$ & $0.4723$
& $2153.5336$ & $0.4723$ 
& $2182.0722$ & $0.4661$ \\ \hline
$0.8$ & $1914.4623$ & 
$0.5313$ & $1914.6659$ & $0.5312$ & 
$1939.6434$ & $0.5244$ \\ \hline
$0.7$ & $1675.5663$ & 
$0.6070$ & $1675.7698$ & 
$0.6069$ & $1697.2146$ & $0.5993$ \\ \hline
$0.6$ & $1436.6355$ & $0.7080$ & 
$1436.8389$ & $0.7079$ & $1454.7858$ 
& $0.6992$ \\ \hline
$0.5$ & $1197.6609$ & $0.8493$ & 
$1197.8641$ & $0.8491$ & $1212.3570$ 
& $0.8390$ \\ \hline
$0.4$ & $958.6283$ & $1.0610$ & 
$958.8314$ & $1.0608$ & $969.9282$
 & $1.0487$ \\ \hline
$0.3$ & $719.5143$ & $1.4136$ & 
$719.7172$ & $1.4132$ & $727.4994$ & 
$1.3981$ \\ \hline
$0.2$ & $480.2728$ & $2.1178$ &
 $480.4753$ & $2.1169$ & $485.0706$ & 
$2.0969$ \\ \hline
$0.1$ & $240.7851$ & $4.2242$ & 
$240.9869$ & $4.2207$ & $242.6418$ 
& $4.1919$ \\ \hline
\end{tabular}
\caption{Performance comparison}
\label{tab:tabcomp}
\end{table}

\begin{figure}[bht]
\centerline{\psfig{file=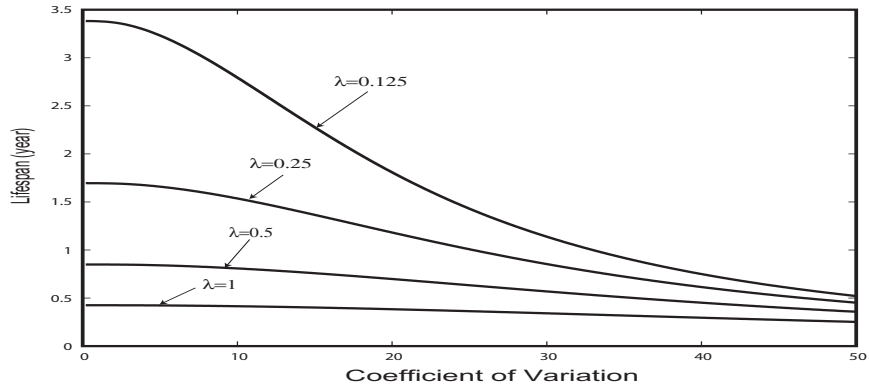,height=2in,width=4.5in}}
\caption{Life span vs $c_v(y)$
$[b_{size}(128b),\mu_y(128b),e^{wu}_w(80\mu J),p^{idle}_m(0.409\mu W)]$}
\label{fig:year}
\end{figure}

To see the size variation effect 
on the fixed-size buffering scheme, 
we consider a
hypothetical symmetrical size 
distribution ($\gamma(y)=0$). Thus
the optimal buffer size $b^*$ is expressed as
\begin{align} \label{equ:optimalll}
b^*=
\sqrt{\frac{2 \lambda e^{wu}_m 
\mu_y b_{size}}{p^{idle}_m}+ 
\mu_y^2 (\frac{5 c^4_v(y)}{4} + \frac{1}{12})}
\end{align}
Based on \eqref{equ:optimal-size} and  Table~2,
Figure\,\ref{fig:year} plots the lifespan as a function of $c_v(y)$
under different data arrival rates 
($\lambda$). One can see
a monotonic decline in the lifespan with increasing $c_v(y)$. This
trend indicates that for
the {\it fixed-size} buffering scheme, high
data size variation has a detrimental effect that
further depletes battery.
By contrast, 
the fixed-interval buffering scheme 
has an obvious advantage of 
being eminently immune to data size variation:
its power consumption
is only associated to mean data size $\mu_y$,
and is independent of data size variance. 
To illustrate, under two power parameter settings,
the lifespan, optimal buffer interval $T^*$, and 
the power consumption under different rates are given in Table~\ref{tab:tab5}.
\begin{table}[htb]
\centering
\begin{tabular}{|c|c|c|c||c|c|c|} \hline
$\lambda (1/s)$& $T^* (s)$& 
$\overline{e^{FI}(T^*)} (\mu W)$& lifespan (yr) & $T^*$ (s)& 
$\overline{e^{FI}(T^*)} (\mu W)$& lifespan (yr) \\ \hline
$1$ & $19.79$ & $2392.38$ & $0.425$ & $12.65$ & $3230.78$ & $0.315$ \\ \hline
$0.5$ & $13.99$ & $4780.02$ & $0.213$ & $8.94$ & $6387.46$ & $0.159$ \\ \hline
$0.25$ & $9.89$ & $9553.33$ & $0.107$ & $6.32$ & $12670.13$ & $0.08$ \\ \hline
$0.125$ & $6.99$ & $19097.18$ & $0.053$  & $4.47$ & $25192.07$ & $0.04$ \\ \hline
$0.0625$ & $4.94$ & $38180.97$ & $0.027$ & $3.16$ & $50174.57$ & $0.02$ \\ \hline
power& 
\multicolumn{3} {c||}{$e^{wu}_w(80\,\mu J)$, $p^{idle}_m(0.409\,\mu W)$} 
& \multicolumn{3} {|c|}{$e^{wu}_w(800\,\mu J)$, $p^{idle}_m(10\,\mu W)$} \\ 
parameters& \multicolumn{3} {c||}{$e^{TX}_w(8.976\,\mu J/byte)$} 
& \multicolumn{3} {|c|}{$e^{TX}_w(8.976\,\mu J/byte)$} \\ \cline{2-7}
& \multicolumn{6} {c|}{
$b_{size}(128b), \mu_y(128b), e^r_m+e^w_m(36 \cdot 10^{-3}\mu J/byte)$}\\ \hline
\end{tabular}
\caption{Lifespan under data arrival rates for fixed-interval
buffering
scheme} \label{tab:tab5}
\end{table}

For environmental monitoring, a sensor-based network
forms a data collection tree.
Each node gathers
local information, and forwards the 
data from its child
nodes to its routing parent
(see Figure\,\ref{fig:childparent}).
The sink node then collates the received
information into global environmental data.
Below we establish an energy
consumption relationship between a routing 
parent node and its child nodes in 
the context of fixed-interval buffering scheme. 

\begin{figure}[tbh] 
\centerline{\psfig{file=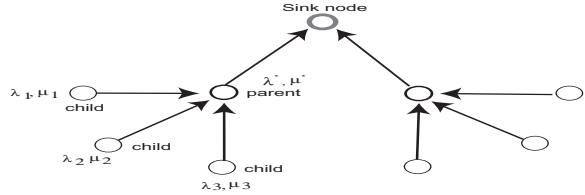,height=1in,width=3in}}
\caption{Data collection tree: bandwidth conservation between parent 
and child nodes}
\label{fig:childparent}
\end{figure}

\begin{theorem} \label{theo:relay}
Suppose a parent node has $k$ child nodes in a static data collection tree.
Let the {\it i}th child have Poisson arrival with a rate of $\lambda_i$ and a mean size of
$\mu_i$. If the parent and child nodes 
adopt
the fixed-interval buffering scheme, then
the optimal buffer interval $T_p^*$ for the parent node is
\begin{align} \label{equ:parentnode}
T^*_p = \sqrt{\frac{2 e^{wu}_w b_{size}}{p^{idle}_m
\sum\limits_{i=1}^k \lambda_i \mu_i}}.
\end{align}
\end{theorem}

\begin{proof} It follows from 
\eqref{equ:op-fixed-interval} that
for the {\it i}th child node, 
its optimal interval is
$T_i = \sqrt{\frac{2 e^{wu}_w b_{size}}{p^{idle}_m \lambda_i
\mu_i}}$. This means that successive 
transmission from the {\it i}th 
node to its parent node is equally spaced 
by an interval $T_i$, with the
mean size of $\lambda_i T_i \mu_i$. 
Assume that each node has the radio-triggered 
wakeup capability, thereby incurring no listening power consumption.  

Let $T$ be the length
of a renewal cycle of the parent node.
Each child node
independently transmits the buffered data at
the rate of $1/T_i$. With respect to the {\it i}th child, 
the corresponding energy consumed by the parent 
node, denoted by $e^{FI}(T)(i)$, 
can be decomposed into three pieces: 

1) The energy for data buffering (idle-mode)
at the parent node, denoted by $e^{FI}_b(T)(i)$,
is bounded as:
\begin{align} \label{equ:parentbuffer}
\sum_{k=1}^{\lfloor \frac{T}{T_i}\rfloor} \dfrac{p_m^{idle}(T-k
T_i) y_i(k)}{b_{size}} \leq e^{FI}_b(T) \leq
\sum_{k=0}^{\lfloor \frac{T}{T_i}\rfloor} 
\dfrac{p_m^{idle}(T-k T_i) y_i(k)}{b_{size}},
\end{align}
where $y_i(k)$ is the size of {\it k}th data
sent by the {\it i} child node, and
$E[y_i(k)]=\lambda_i T_i \mu_i$. 
\begin{align}
 \sum_{k=1}^{\lfloor \frac{T}{T_i}\rfloor} \dfrac{p_m^{idle}(T-k
T_i) \lambda_i T_i \mu_i}{b_{size}} \leq E[e^{FI}_b(T)] \leq
 \sum_{k=0}^{\lfloor \frac{T}{T_i}\rfloor} \dfrac{p_m^{idle}(T-k
T_i) \lambda_i T_i \mu_i}{b_{size}}.
\end{align}
This yields $E[e^{FI}_b(T)(i)] 
\approx \frac{p_m^{idle} T^2}{2T_i}\frac{\lambda_i T_i
\mu_i}{b_{size}} = \frac{p_m^{idle}T^2 \lambda_i \mu_i}{2b_{size}}$. \\

2) The energy for data transmission and 
for reading/writing data from/into memory banks, plus 
elevating/demoting the power status 
of memory banks,
is 
\begin{align}
& e^{FI}_t(i)= \sum_{k=0}^{n_i(T)} 
y_i(k) (e^{TX}_w+ e^w_m + e^r_m)+2e^{resyn}_m,
\\ \nonumber
& E[e^{FI}_t(i)]= T \left(\lambda_i \mu_i
 (e^{TX}_w +e^w_m+ e^r_m)+2e^{resyn}_m/T_i\right), 
\end{align}
where $n_i(T)$ is the number of 
transmissions by the {\it i}th node over 
$T$ interval, 
and its expectation $E[n_i(T)]$ is $T/T_i$. \\

3) The energy for data reception ($e^{FI}_r(i)$) 
from the {\it i}th child is 
\begin{align}
e^{FI}_r(i)=
\sum_{k=0}^{n_i(T)} y_i(k) e^{RX}_w + 
e^{wu}_w,  \ \ \
E[e^{FI}_r(i)]=
T \lambda_i \mu_i e^{RX}_w + \frac{ T e^{wu}_w}{T_i}, 
\end{align}
where $n_i(T)$ denotes the
number of data receptions 
at the parent over $T$, so that the number of 
radio-wakeups
(by a radio-triggered wakeup mechanism) is $T/T_i$, and  
the expected energy in radio wakeup over 
$T$ is $\frac{T e^{wu}_w}{T_i}$.
The energy for data reception
is $T\lambda_i \mu_i e^{RX}_w \approx T\lambda_i 
\mu_i e^{TX}_w$, 
assuming that $e^{RX}_w \approx e^{TX}_w$.

Thus the average total 
energy of the parent node with
$k$ child nodes over $T$ is:
\begin{align} \label{equ:parentc1}
E[e^{FI}(T)] = e^{wu}_w + \dfrac{p^{idle}_m T^2}{2 b_{size}}
\sum_{i=1}^{k}\lambda_i \mu_i 
+ T \sum\limits_{i=1}\limits^k 
\left (\lambda_i \mu_i (2e^{TX}_w+
e^w_m +e^r_m)+ \frac{2 e^{resyn}_m + e^{wu}_w }{T_i}\right)
\end{align}
Notice that only one radio wakeup 
for data transmission and 
$\sum_{i=1}^{k} T/T_i$ radio wakeups 
for data receptions from $k$ child nodes
in each cycle $T$.
Using the same trick 
in proof of Theorem~\ref{the:the2}, the
long-run mean average energy is 
\begin{align} \label{equ:balance}
\overline{e^{FI}(T)} = \frac{e^{wu}_w}{T} + 
\dfrac{T
p^{idle}_m}{2 b_{size}}
\sum_{i=1}^{k}\lambda_i \mu_i 
+ 
\sum_{i=1}^k \left (\lambda_i \mu_i 
(2e^{TX}_w+
e^w_m +e^r_m)+ \frac{2 e^{resyn}_m+e^{wu}_w}{T_i}\right).
\end{align}
Taking derivative of (\ref{equ:balance}) w.r.t. $T$ gives
$\dfrac{ \partial \overline{e^{FI}(T)}}{\partial T} =
-\dfrac{e^{wu}_w}{T^2}+ \dfrac{\lambda p^{idle}_m}{2b_{size}}
\sum\limits_{i=1}\limits^{k} \lambda_i \mu_i$.
Resolving $\frac{\partial \overline{e^{FI}(T_p)}}{\partial T_p} =0$
yields (\ref{equ:parentnode}). $\hfill$
\end{proof} 

Let $v=(\lambda_1\mu_1,\cdots,\lambda_k\mu_k)$ be a
bandwidth distribution vector of $k$ child nodes,
and $\overline{e^{FI}(T_p)}(v)$ refer to the long-run mean average energy consumption
of the parent node under $v$.
Substituting
\eqref{equ:parentnode} into \eqref{equ:balance} gives
\begin{align} \label{equ:optimal-fixed-interval1}
\overline{e^{FI}(T_p)}(v) &= \sqrt{\dfrac{2 p^{idle}_m e^{wu}_w
\sum\limits_{i=1}\limits^k \lambda_i \mu_i}{b_{size}}} 
+(2e^{TX}_w +e^w_m+e^r_m)\sum\limits_{i=1}\limits^{k} 
\lambda_i
\mu_i \nonumber \\  
& + (e^{resyn}_m+\frac{e^{wu}_w}{2}) \sqrt{\frac{2 p^{idle}_m}{e^{wu}_w b_{size}}}
\sum\limits_{i=1}\limits^{k} \sqrt{\lambda_i \mu_i},
\end{align}
where
$\sum\limits_{i=1}^k \lambda_i \mu_i$
refers to the total bandwidth of the parent node. \\

{\it A natural question arises how bandwidth distribution among
child nodes affects the overall power consumption of the parent
node.} \\

To answer this question, we first introduce the notion of {\it
majorization}, and then provide a lemma to facilitate necessary
derivations.

For any vector $x\!=\!(x_1,\cdots,x_n)\in {\cal R}^n$, let $x_{(1)}\!\leq\!\cdots
\leq x_{(n)}$ be the component of $x$ in ascending order, 
and $x_{\downarrow}\!=\!(x_{(1)},\cdots,x_{(n)})$ be
the ascending rearrangement of $x$. 

\begin{definition} \label{def:maj}
For two vectors $x,y \in {\cal R}^n$,
\begin{align}
x \prec y \  \  \  \mbox{if} \ \
\begin{cases}
\sum\limits_{i=1}^k x_{(i)} \geq \sum\limits_{i=1}^k y_{(i)}, &
1\leq k \leq n-1 \\
\sum\limits_{i=1}^n x_{(i)}=\sum\limits_{i=1}^n y_{(i)} &
\end{cases}
\end{align}
Then $x$ is said to be {\it majorized} by $y$  [\citeNP{Marshall1979}].
\end{definition}

A trivial example below is given to illustrate the notion of {\it
majorization}:
\begin{align}
(\frac{1}{n},\cdots,\frac{1}{n}) 
\prec(\frac{1}{n-1},\cdots,\frac{1}{n-1},0) \prec \cdots \prec(\frac{1}{2},\frac{1}{2},0,\cdots,0) \prec(1,0,\cdots,0)
\nonumber
\end{align}

\begin{lemma} \label{lem:lemma-maj}
Let $x=(x_1,\cdots x_n),y=(y_1,\cdots,y_n)$ be two vectors and $g$
be a concave function. If $x \prec y$, then
$\sum\limits_{i=1}^n g(y_i) < \sum\limits_{i=1}^n g(x_i)$.
\end{lemma}
The proof can be seen in [\citeNP{Marshall1979}] and therefore is
omitted. 

\begin{theorem} \label{theo:maj}
Let $v\!=\!(\lambda_1\mu_1,\cdots,\lambda_k \mu_k)$
and $v^\prime\!=\!(\lambda_1^\prime \mu_1^\prime,\cdots,\lambda_k^\prime
\mu_k^\prime)$ be two child bandwidth distribution vectors.
Letting $\sum_{i=1}^k \lambda_i \mu_i = \sum_{i=1}^k
\lambda_i^\prime \mu_i^\prime={\cal B}$. If $v$ is majorized
by $v^\prime$ ($v\!\prec\!v^\prime$), then
the parent node consumes more power
under $v$ than under $v^\prime$,
with the same optimal buffer interval $T_p$.  Namely,
$\overline{e^{FI}(T_p)}(v)> \overline{e^{FI}(T_p)}(v^\prime)$. 
\end{theorem}

\begin{proof}
Since
$\sum_{i=1}^k \lambda_i \mu_i\!=\!\sum_{i=1}^k\lambda_i^\prime \mu_i^\prime\!=\!{\cal B}$,
by (\ref{equ:parentnode}),
the optimal 
interval for the parent node,
$T_p=\sqrt{\frac{2 e^{wu}_w b_{size}}{p^{idle}_m {\cal B}}}$,
is identical under both $v$ and $v^\prime$.
Based on lemma~\ref{lem:lemma-maj} that 
$\sum_{i=1}^k \sqrt{\lambda_i \mu_i} > \sum_{i=1}^k
\sqrt{\lambda_i^\prime \mu_i^\prime}$
since $v \prec v^\prime$, we get
\begin{align} \label{equ:dif}
\overline{e^{FI}(T_p)}(v)-\overline{e^{FI}(T_p)}(v^\prime) 
= (e^{resyn}_m+\frac{e^{wu}_w}{2}) \sqrt{\frac{2p^{idle}_m}{e^{wu}_w b_{size}}}
\left (\sum_{i=1}\limits^{k}
\sqrt{\lambda_i \mu_i}\!-\!\sqrt{\lambda_i^\prime \mu_i^\prime} \right )> 0.
\end{align}
Theorem~\ref{theo:maj} is thus proved. $\hfill$
\end{proof}

Let ${\cal V}_B$ be a {\em convex} space formed by a set of vectors satisfying
$v=(\lambda_1\mu_1,\cdots,\lambda_k\mu_k) \in {\cal V}_B$ iff
$\sum\limits_{i=1}^k \lambda_i \mu_i\!=\!{\cal B}$.
Let
$v_{\leftrightarrow}\!=\!(\frac{{\cal B}}{k},\cdots,\frac{{\cal B}}{k}) \in {\cal V}_B$
and $v_{\updownarrow}\!=\!({\cal B},0,\cdots,0) \in {\cal V}_B$
be two child bandwidth distribution vectors.
The vector $v_{\leftrightarrow}$ represents a
uniform bandwidth distribution in which
each child equally contributes the $\frac{{\cal B}}{k}$ {\em bandwidth}
of the parent, while $v_{\updownarrow}$ is
an extremely uneven bandwidth distribution
where only child constitutes 
the ${\cal B}$ {\em bandwidth} of the
parent node and the remaining children contribute nothing. 
Clearly, for any bandwidth distribution
$v\in {\cal V}_B$, $v_{\leftrightarrow}\!\prec v\!\prec v_{\updownarrow}$. 
It follows from Theorem~\ref{theo:maj}, we have
\begin{align}
\overline{e^{FI}(T_p)}(v_{\leftrightarrow})=\max
\{\overline{e^{FI}(T_p)} (v): v \in {\cal V}_B \} \\
\overline{e^{FI}(T_p)}(v_{\updownarrow}) =
\min \{\overline{e^{FI}(T_p)} (v):  v \in {\cal V}_B \}.
\end{align}
By the aid of \eqref{equ:dif}, 
the power difference
of the parent node under $v_{\leftrightarrow}$ and $v_{\updownarrow}$ is
\begin{align} \label{equ:impactrange}
 \overline{e^{FI}(T_p)}(v_{\leftrightarrow})
-\overline{e^{FI}(T_p)}(v_{\updownarrow}) =
(e^{resyn}_m+\frac{e^{wu}_w}{2}) \sqrt{\frac{2p^{idle}_m {\cal B}}{e^{wu}_w b_{size}}}
(\sqrt{k}-1).
\end{align}

Theorem~\ref{theo:maj} provides a means of
quantifying the impact of uniformity in 
child bandwidth distribution
on power consumption of the parent node.
\eqref{equ:impactrange} in particular gives
the bound on the range of such
bandwidth distribution effect,
which is proportional to the square root of
the number of child nodes. 

\section{Conclusion and Future Work}
The longevity of battery-powered sensor networks is an essential
performance metric of around-clock environmental surveillance and
monitoring. 
This paper focuses on the exploitation of
power-aware buffering schemes to reduce power consumption of
sensor networks based on the radio-triggered 
power management. It shows insofar as that
the power-aware buffering is a non-negligible 
factor that effectively improve the
lifespan of sensor networks, and that 
the power-oblivious buffering is harmful as
it is very likely to result in an 
excessive power consumption. 

An in-depth analysis shows that 
the fixed-size and fixed-interval
buffering schemes differ markedly in relation to 
data size variability. The
power-aware fixed-size buffering scheme is implicated in both the
skewness and coefficient of variation in the data size distribution,
and its performance 
could deteriorate rapidly when the data size
is of high-variance. In contrast, the
hallmark of the fixed-interval buffering scheme is its
immunity to the data-size variation. 
The fixed-interval buffering scheme is 
therefore the buffering choice for its
performance stability in a variety of sensor-based 
application environments. 
Furthermore, in the context of the 
fixed-interval buffering scheme, we
establish the power consumption 
relationship between parent and
child nodes in a static data collection tree in sensor networks.
We show that a uniform
bandwidth distribution among child nodes in fact consumes 
more power of the parent node than an uneven bandwidth distribution.

These findings are valuable in understanding the asymptotic behavior
of the power-aware buffering schemes in the presence of size
variability. They provide
well-informed guidance on determining the optimal buffer size or buffer
interval based on the power parameter of radio and memory banks,
allowing us to judiciously select a buffering scheme that better tailors 
to data arrival rate and data size distribution.

Our future work will focus on 1) validating the 
buffering models in a lab environment, 
including simulation and system implementation;
2) studying the power-aware buffering issue
under real-time constraints. The goal of the new
research avenue is to study strategy that 
can provide optimal 
trade-off between power-aware buffering and 
responsiveness.     

\section{Acknowledgements}
We would like to thank the anonymous reviewers 
for their incredibly insightful critiques that 
help us to significantly improve the quality of the paper. 

\bibliographystyle{acmtrans}
\bibliography{sensor}
\section*{Appendix}
We first introduce the notion of
first ladder height, then present Lai and
Seigmund's remarkable theorem [\citeNP{Lai1977}], which has laid
theoretical basis for quantifying the size variation impact
on the fixed-size buffering scheme.  

\begin{definition}
Let $x_1, x_2, \cdots$ be random variables
following certain distribution $\cal{F}$ with parameter $\theta$, and
$S_n$
be the random walk consisting of
the partial sum $S_n=\sum_{i=1}^n x_i$.
The first time $\tau = inf\{n: S_n >0\}$ that
the random walk is positive is called the
{\it first ladder epoch} and the first positive value
$S_r$ taken by the random walk is
called as the {\it first ladder height}. 
\end{definition}
A detailed derivation of Theorem~\ref{theo:Lai1979}
can be found in [\citeNP{Lai1977};\citeyearNP{Lai1979}]. 
\citeN{Keener1987} later gave
a simplified expression for $k$
in terms of  moments of ladder height variables. 

\begin{theorem} \label{theo:Lai1979}
Let $\{x_i, i \geq 1 \}$
be a random walk with a mean of $\mu_x=E [x_1] > 0$ and finite
variance $\sigma_x^2$.
Let $\tau(b) =
\min \{ n \geq 1: \sum\limits^{n}\limits_1 x_i > b\}$,
$R_b=S_{\tau(b)}-b,M\!=\!\min\limits_{n\geq 0}\!S_n,\tau^+=\!\tau(0)$,
and $H=S_{\tau^+}$.
As $b\!\rightarrow\!\infty$,
the stopping time variance becomes
\begin{align} \label{equ:key}
\sigma^2_{\tau(b)} = \dfrac{b\sigma_x^2}{\mu_x^3} + \dfrac{k}{\mu_x^2} + o(1),
\end{align}
\end{theorem}
where $\tau(b)$ refers to the stopping time and
$k$ is the key constant with a
rather complicated expression as follows:
\begin{align} \label{equ:lai}
k=\dfrac{\sigma_x^2 EH^2}{2\mu_x EH}
+\dfrac{3}{4}\left(\dfrac{EH^2}{EH}\right)^2
-\dfrac{2}{3}\dfrac{EH^3}{EH}-\dfrac{EH^2 EM}{EH}
-2\int\limits_0^{\infty} E R_x P(M \leq -x) dx
\end{align}

Based on Theorem~\ref{theo:Lai1979}, we provide a proof of Theorem~1 in this paper 
below. 

\begin{proof}  Consider random
variables with positive increments $\{ x_i > 0, i \geq 0 \}$ because
the received data size is always positive. Let's reexamine the expression for
$k$ in (\ref{equ:lai}) under the positive increment
condition. Define the ladder epochs $\tau_0 = 0$ and $\tau_1^+ = \tau^+$,
and $\tau_{n+1}^+\!=\! inf\{k >\tau_n^+:S_k>S_{\tau_n^+}\}$ for
$n\!\geq\!1$. For $n>0$, the (n+1)th ladder height
$H_{n+1}= S_{\tau_{n+1}^+}\!-\!S_{\tau_{n}^+}$.
Consider $H_{n+1}= S_{\tau_{n+1}^+} -
S_{\tau_{n}^+} =x_{n+1}$ under the positive increment condition,
{\it i.e.}, $H=x$. As a result, we obtain
\begin{align}
& EH=E[x]=\mu_x,
EH^2=E[x^2] = \sigma^2_x + \mu^2_x, \\ \nonumber
& EH^3=E[x^3]=3\sigma^2_x\mu_x +\mu_x^3+\sigma^3_x \gamma(x)
\end{align}
Using the identity $EM\!=\!\dfrac{EH^2}{2 EH}-\dfrac{E
[x^2]}{2E[x]}$ [\citeNP{Woodroofe1976}], 
we obtain that $EM =
0$ when the positive random walk is assumed. Substitution of these
identities into (\ref{equ:lai}) yields a simplified expression for
$k$, denoted by $k^*$, as follows:
\begin{align} \label{equ:final1}
k^*&=
\dfrac{\sigma^2_x (\sigma^2_x +\mu^2_x)}{2 \mu_x^2}
+\dfrac{3}{4} \left (\dfrac{\sigma^2_x +\mu^2_x}{\mu_x} \right)^2 
-\dfrac{2}{3} \left (
 \dfrac{3 \sigma^2_x \mu_x + \mu_x^3 + \sigma^3_x \gamma(x)}
{\mu_x} \right) \\ \nonumber
&= \dfrac{5 \sigma^4_x}{4 \mu^2_x}+
\dfrac{\mu_x^2}{12}-\dfrac{2 \sigma^3_x \gamma(x)}
{3 \mu_x}= \dfrac{5 \sigma^2_x c^2_v(x)}{4}+
\dfrac{\mu_x^2}{12}-\dfrac{2 \sigma^3_x c_v(x) \gamma(x)}
{3},\
\end{align}
Substitution of (\ref{equ:final1})
into (\ref{equ:key}) completes the proof. $\hfill$
\end{proof} 

\end{document}